\documentclass{article}
\usepackage{amssymb}
\usepackage{amsfonts}
\usepackage{amsmath}
\usepackage{graphicx,epsfig}
\usepackage[enableskew]{youngtab}

\newtheorem{theorem}{Theorem}

\newtheorem{corollary}[theorem]{Corollary}

\newtheorem{definition}[theorem]{Definition}
\newtheorem{example}[theorem]{Example}

\newtheorem{lemma}[theorem]{Lemma}

\newtheorem{proposition}[theorem]{Proposition}
\newtheorem{remark}[theorem]{Remark}

\newenvironment{proof}[1][Proof]{\noindent\textbf{#1.} }{\ \rule{0.5em}{0.5em}}
\numberwithin{equation}{section}
\numberwithin{theorem}{section}
\input{tcilatex}

\begin{document}

\title{From Quantum B\"{a}cklund Transforms to Topological Quantum Field
Theory}
\author{Christian Korff \\
{\normalsize School of Mathematics and Statistics, University of Glasgow}}
\maketitle

\begin{abstract}
We derive the quantum analogue of a B\"{a}cklund transformation for the
quantised Ablowitz-Ladik chain, a space discretisation of the nonlinear Schr%
\"{o}dinger equation. The quantisation of the Ablowitz-Ladik chain leads to
the $q$-boson model. Using a previous construction of Baxter's Q-operator
for this model by the author, a set of functional relations is obtained
which matches the relations of a one-variable classical B\"{a}cklund
transform to all orders in $\hbar $. We construct also a second $Q$-operator
and show that it is closely related to the inverse of the first. The multi-B%
\"{a}cklund transforms generated from the $Q$-operator define the fusion
matrices of a 2D TQFT and we derive a linear system for the solution to the
quantum B\"{a}cklund relations in terms of the TQFT fusion coefficients.
\end{abstract}

\section{Introduction}

In the context of classical integrable systems the main interest in the
construction of B\"{a}cklund and Darboux transformations is their
application in the construction of soliton solutions \cite{RS82,MS91}.
Compared to the classical B\"{a}cklund transformations the discussion of
their quantum cousins started more recently. In 1992 Gaudin and Pasquier 
\cite{PG92} constructed for the quantum Toda chain an analogue of Baxter's $%
Q $-operator \cite[Ch. 9-10]{Baxter07} and showed that in the semi-classical
limit $\hbar \rightarrow 0$ the similarity transformation $\mathcal{O}%
\rightarrow Q(u)\mathcal{O}Q(u)^{-1}$ is a B\"{a}cklund transform of the
Toda chain. An introductory account to this and further results can be found
in Sklyanin's lecture notes \cite{Sklyanin00}. 

\subsection{The Ablowitz-Ladik chain and its quantisation}

In this article we are interested in the quantum analogue of B\"{a}cklund
transforms for another integrable system: the Ablowitz-Ladik chain \cite%
{AL76}, 
\begin{equation}
\left\{ 
\begin{array}{c}
\partial _{t}\psi _{j}=\psi _{j+1}-2\psi _{j}+\psi _{j-1}-\psi _{j}^{\ast
}\psi _{j}(\psi _{j+1}+\psi _{j-1}) \\ 
\partial _{t}\psi _{j}^{\ast }=-\psi _{j+1}^{\ast }+2\psi _{j}^{\ast }-\psi
_{j-1}^{\ast }+\psi _{j}^{\ast }\psi _{j}(\psi _{j+1}^{\ast }+\psi
_{j-1}^{\ast })%
\end{array}%
\right. ,  \label{ALH}
\end{equation}%
where we consider periodic boundary conditions $\psi _{j+n}=\psi _{j}$ and $%
\psi _{j+n}^{\ast }=\psi _{j}^{\ast }$. The equations (\ref{ALH}) are a
space discretisation of the following system of coupled PDEs%
\begin{equation*}
\left\{ 
\begin{array}{c}
\partial _{t}\psi =\partial _{x}^{2}\psi -2\psi ^{\ast }\psi ^{2} \\ 
\partial _{t}\psi ^{\ast }=-\partial _{x}^{2}\psi ^{\ast }+2\psi ^{\ast
2}\psi%
\end{array}%
\right. \;.
\end{equation*}%
After changing to imaginary time, $t\rightarrow \imath t$ with $\imath =%
\sqrt{-1}$, this system of PDEs allows for a reduction, $\psi ^{\ast }=\pm 
\bar{\psi}$ with $\bar{\psi}$ denoting the complex conjugate of $\psi $, to
the nonlinear Schr\"{o}dinger (NLS) equation, $-\imath \partial _{t}\psi
=\partial _{x}^{2}\psi \mp 2\psi |\psi |^{2}$. B\"{a}cklund and Darboux
transformations for the system (\ref{ALH}) and its reduction to the discrete
NLS system can be found in e.g. \cite{CM83,PBL96,Suris97,Vekslerchik06} and
references in \emph{loc. cit.}

Kulish considered in \cite{Kulish81} a particular quantisation of the
Ablowitz-Ladik chain leading to the $q$-boson model with Hamiltonian%
\begin{equation}
H=-\sum_{j=1}^{n}(\beta _{j}\beta _{j+1}^{\ast }+\beta _{j}^{\ast }\beta
_{j+1}-2(1-q^{2})N_{j}),\;  \label{Hqboson}
\end{equation}%
see also \cite{BB92,BBP93,BIK98} as well as references therein. Here the $%
\beta _{j}$, $\beta _{j}^{\ast }$'s are the generators of a $q$-deformed
version of the oscillator or Heisenberg algebra with $q\ $being the
quantisation parameter (see the definition (\ref{qbosondef}) in the text)
and $N_{j}$ is the particle number operator at site $j$.

In the Fock space representation and with periodic boundary conditions on
the lattice, $\beta _{j+n}=\beta _{j}$ and $\beta _{j+n}^{\ast }=\beta
_{j}^{\ast }$, the model (\ref{Hqboson}) can be solved via the quantum
inverse scattering method \cite{FST79}. Analogous to the classical case, one
defines a quantum Lax operator $L(u)$ and the Hamiltonian (\ref{Hqboson})
can then be understood as a particular element in a commutative algebra
generated from the commuting transfer matrices $T(u)=\limfunc{Tr}L(u)$ of an
exactly solvable lattice model in the sense of Baxter \cite{Baxter07}; see
Figure \ref{fig:Tvertexmodel} for the vertex configurations defining the
model.

\begin{figure}[tbp]
\begin{equation*}
\includegraphics[scale=0.27]{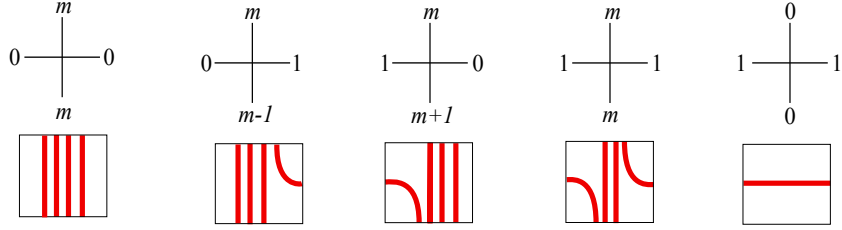}
\end{equation*}%
\caption{Vertex configurations for the $T$-operator. Consider a square
lattice where we assign a nonnegative integer to each lattice edge, the
statistical variable. At each vertex we only allow for particular
configurations, for the $T$-operator the horizontal edges can only take
values 0 or 1, below are the vertex configurations in terms of
non-intersecting paths.}
\label{fig:Tvertexmodel}
\end{figure}

\subsection{Baxter's Q-operators and quantum B\"{a}cklund maps}

In this article we discuss two analogues $Q^{\pm }$ of Baxter's $Q$-operator
for the quantised Ablowitz-Ladik chain (\ref{Hqboson}) with periodic
boundary conditions. The construction of these $Q^{\pm }$-operators
corresponds to the definition of two additional exactly solvable statistical
mechanics models defined on a square lattice; the corresponding vertex
configurations are depicted in Figure \ref{fig:Qvertexmodels}. We show that
the transfer matrices $T(u),Q^{\pm }(u)$ of all three models are related via
functional relations, Baxter's famous $TQ$-equation and a quantum Wronskian
relation for $Q^{\pm }$. The $Q^{+}$-operator has been constructed earlier
by the author, the construction of $Q^{-}$ and the resulting functional
relations (\ref{TQ-}), (\ref{qWronskian}) in the text are new. We show that
in the limit $n\rightarrow \infty $ the second solution $Q^{-}$ becomes the
inverse of $Q^{+}$.

\begin{figure}[tbp]
\begin{equation*}
\includegraphics[scale=0.27]{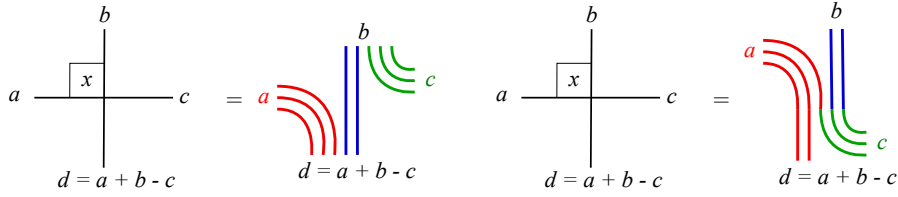}
\end{equation*}%
\caption{Vertex configurations for the $Q^\pm$-operators. Consider a square
lattice where we assign a nonnegative integer to each lattice edge, the
statistical variable. At each vertex we only allow for particular
configurations, for the $Q^+$-operator (bottom left) we require that $b\geq
c $ and $a+b=c+d$. For $Q^-$ (bottom right) we impose instead the conditions 
$a+b\geq c$ and $a+b=c+d$. On the right of each vertex is an interpretation
of these conditions in terms of non-intersecting lattice paths.}
\label{fig:Qvertexmodels}
\end{figure}

Employing the $Q^{+}$-operator we consider the associated similarity
transformations of the $q$-boson fields,%
\begin{equation*}
\beta _{j}\mapsto \tilde{\beta}_{j}(v)=Q^{+}(v)\beta _{j}Q^{+}(v)^{-1}\quad 
\text{and}\quad \beta _{j}^{\ast }\mapsto \tilde{\beta}_{j}^{\ast
}(v)=Q^{+}(v)\beta _{j}^{\ast }Q^{+}(v)^{-1}\;.
\end{equation*}%
The first main result of this article is the proof that the transformed
fields $\tilde{\beta}_{j}(v)$ and $\tilde{\beta}_{j}^{\ast }(v)$ obey the
functional relations (c.f. Theorem \ref{mainresult} in the text)%
\begin{equation}
\frac{\tilde{\beta}_{j}-\beta _{j}}{v}=\left( 1-\beta _{j}^{\ast }\tilde{%
\beta}_{j}\right) \tilde{\beta}_{j-1}\qquad \text{and}\qquad \frac{\tilde{%
\beta}_{j}^{\ast }-\beta _{j}^{\ast }}{v}=\beta _{j+1}^{\ast }(\beta
_{j}^{\ast }\tilde{\beta}_{j}-1),  \label{quantumBintro}
\end{equation}%
which allow one to compute $\tilde{\beta}_{j}(v)$ and $\tilde{\beta}%
_{j}^{\ast }(v)$ recursively via a power series expansion in $v$. This
result is physically significant, since this B\"{a}cklund map (and the one
induced by the adjoint of the $Q^{+}$-operator) describes for small time
steps $0<v=\Delta t\ll 1$ the discrete time evolution of the quantum
Ablowitz-Ladik chain. Surprisingly, the quantum relations (\ref%
{quantumBintro}) match \emph{exactly} (i.e. to all orders in $\hbar $) the
relations of the classical fields $\psi _{j},\psi _{j}^{\ast }$ under the
analogous classical one-variable B\"{a}cklund map $\mathcal{B}^{+}(v):(\psi
_{j},\psi _{j}^{\ast })\mapsto (\tilde{\psi}_{j}(v),\tilde{\psi}_{j}^{\ast
}(v))$ considered by Suris in \cite{Suris97} for the system (\ref{ALH}).

In comparison, the match between Baxter's Q-operator and the classical B%
\"{a}cklund map for the Toda chain is established by formulating $Q$
as an integral operator and then identifying its integral kernel in the
limit $\hbar \rightarrow 0$ with the generating function of the B\"{a}cklund
transform; see \cite{PG92}. Our approach avoids the generating function
altogether and directly arrives at the functional relation (\ref%
{quantumBintro}) which allows one to compute the image of the quantum transform,
the quantum variables $\tilde{\beta}_{j},\,\tilde{\beta}_{j}^{\ast }$, via
recurrence; see Section \ref{sec:result}.

\subsection{Multivariate B\"{a}cklund maps and 2D TQFT}

Once a one-variable B\"{a}cklund map is constructed one can consider
multivariate transforms via composition, $\mathcal{B}^{+}(x_{1})\circ 
\mathcal{B}^{+}(x_{2})\circ \cdots \circ \mathcal{B}^{+}(x_{n-1})$. These
multivariate transforms are central to Sklyanin's separation of variables
approach; see e.g. \cite{KS98} and, for a discussion of the quantum
separation of variables approach, \cite{Niccoli13} as well as references
therein. Here we wish to highlight a novel aspect of these multivariate
transforms: they generate the fusion matrices of a 2D topological quantum
field theory (TQFT) constructed in \cite{CMP13}.

Consider matrix elements of a product of $Q^{+}$-operators which is the
quantum analogue of the above multivariate B\"{a}cklund transform, 
\begin{equation}
\langle \lambda |\tprod_{i=1}^{n-1}Q^{+}(x_{i})|\mu \rangle =\sum_{\nu
}(-1)^{|\nu|}N_{\mu \nu }^{\lambda }(q)P_{\nu }(x_{1},\ldots ,x_{n-1};q)\;,
\label{fusion0}
\end{equation}%
where $\lambda $, $\mu $, $\nu $ are partitions labelling particle
configurations of the $q$-boson model (see Eqn (\ref{can_basis}) in the
text) and the $P_{\lambda }$'s are a special basis in the ring of symmetric
functions, the so-called $q$-Whittaker functions which are a special case of
Macdonald's functions \cite[Ch. VI]{Macdonald}. Since the $Q^{+}$-operators
for different $x_{i}$'s commute, reflecting the analogous property of the
classical B\"{a}cklund transforms, the above expansion is well-defined. The
expansion coefficients $N_{\mu \nu }^{\lambda }(q)$ are polynomials in $q$
with integer coefficients and define the fusion in a 2D TQFT, i.e. they are
the values of the pair of pants cobordism shown in Figure \ref{fig:fusion}
which can be seen as a 2D analogue of a Feynman diagram describing the
fusion of particles in a QFT. The TQFT defined via (\ref{fusion0}) is a $q$%
-deformation of the $SU(n)$-WZNW fusion ring, that is the constant term $%
N_{\mu \nu }^{\lambda }(0)$ in $N_{\mu \nu }^{\lambda }(q)$ is given by the
operator product expansion of primary fields in WZNW conformal field theory.

The second main result of this article is that the matrix elements of the B%
\"{a}cklund transformed $q$-boson fields $\tilde{\beta}_{j}(v)$ and $\tilde{%
\beta}_{j}^{\ast }(v)$ can be computed in terms of the TQFT fusion
coefficients, thus relating the quantum B\"{a}cklund transform of the
Ablowitz-Ladik chain to fusion in a 2D TQFT; see Prop \ref{prop:linear} in
the text.

\begin{figure}[tbp]
\begin{equation*}
\includegraphics[scale=0.45]{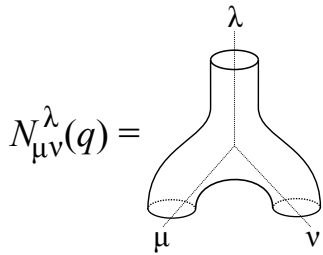}
\end{equation*}%
\caption{Depiction of the `pair-of-pants\rq{} cobordism which fixes the
fusion product in a 2D TQFT. The dotted lines inside are the corresponding
conventional Feynman diagram.}
\label{fig:fusion}
\end{figure}

\subsection{Outline of the article}

The outline of the article is as follows. In Section 2 we review some
aspects of the classical Ablowitz-Ladik chain, in particular we recall a
one-variable family of B\"{a}cklund transforms considered by Suris in \cite%
{Suris97} which is quantised in terms of the $Q^{+}$-operator in Section 4.
Before that Section 3 reviews the necessary background on the quantisation
of the Ablowitz-Ladik chain, the $q$-boson algebra and the algebraic Bethe
ansatz solution of the $q$-boson model. Section 4 gives the construction of
the $Q^{\pm }$-operators and the derivation of Baxter's $TQ$-equation as
well as the resulting quantum Wronskian relation. Section 5 contains the
main result, the discussion of the quantum B\"{a}cklund transform and the
derivation of the functional relations (\ref{quantumBintro}). In Section 6
we then relate these relations to the 2D TQFT defined in terms of the
Q-operator. Section 7 contains a discussion of the results and outlook
towards future work.

\section{The classical Ablowitz-Ladik chain}

Interpreting the discrete fields $\{\psi _{j}\}_{j\in \mathbb{Z}_{n}}$ and $%
\{\psi _{j}^{\ast }\}_{j\in \mathbb{Z}_{n}}$ as components of vectors in $%
\mathbb{R}^{2n},$ the classical Ablowitz-Ladik chain (\ref{ALH}) allows for
the following Poisson structure (see e.g. \cite{Kulish81} and \cite{Suris97}%
) 
\begin{equation}
\{\psi _{i},\psi _{j}^{\ast }\}=\delta _{ij}(1-\psi _{j}^{\ast }\psi
_{j})\qquad \text{and}\qquad \{\psi _{i},\psi _{j}\}=\{\psi _{i}^{\ast
},\psi _{j}^{\ast }\}=0\;.  \label{Poisson}
\end{equation}%
Note that the discrete fields $\psi _{j}$, $\psi _{j}^{\ast }$ are not a
pair of canonical variables. The time evolution (\ref{ALH}) can be described
as Hamiltonian flows, $\partial _{t}\psi _{j}=\{H,\psi _{j}\}$ and $\partial
_{t}\psi _{j}^{\ast }=\{H,\psi _{j}^{\ast }\}$, where (assuming $|\psi
_{j}^{\ast }\psi _{j}|<1$) 
\begin{equation}
H=-\sum_{j\in \mathbb{Z}_{n}}\left( \psi _{j}^{\ast }\psi _{j+1}+\psi
_{j}\psi _{j+1}^{\ast }+2\ln (1-\psi _{j}^{\ast }\psi _{j})\right) \;.
\label{clH}
\end{equation}%
Initially, one can consider $\psi _{j}$ and $\psi _{j}^{\ast }$ as
independent variables, but as explained in the introduction, the system
reduces to the NLS model when changing to imaginary time, $t\rightarrow
\imath t$, and identifying $\psi _{j}^{\ast }$ as the complex conjugate of $%
\psi _{j}$.

The integrals of motion of the system (\ref{ALH}) are obtained from the
spectral curve $\det (L(u)-\lambda )=0$, where the Lax operator $L(u)$ is
given as a product of the following local Lax matrices,%
\begin{equation}
L(u)=L_{n}(u)\cdots L_{2}(u)L_{1}(u),\qquad L_{j}(u)=\left( 
\begin{array}{cc}
1 & u\psi _{j}^{\ast } \\ 
\psi _{j} & u%
\end{array}%
\right) \;.  \label{classicL}
\end{equation}%
Here we have chosen not to take the standard form of the Lax matrices, see
e.g. \cite{AL76,Kulish81,Suris97} and references therein, but instead match
the conventions in \cite{CMP13} which turn out to be convenient to discuss
the quantum case. There are two spectral invariants, the determinant 
\begin{equation}
\det L(u)=u^{n}\prod_{j=1}^{n}(1-\psi _{j}^{\ast }\psi _{j})  \label{detL}
\end{equation}%
and the trace%
\begin{equation}
T(u)=\limfunc{Tr}L(u)=\sum_{r=0}^{n}u^{r}T_{r}\;.  \label{TrL}
\end{equation}%
Using the classical analogue of the Yang-Baxter equation one proves the
following \cite{Kulish81}:

\begin{proposition}
\label{prop:integrable}$\{T_{r},T_{r^{\prime }}\}=0$ for all $r,r^{\prime
}=0,1,\ldots ,n$.
\end{proposition}

The integrals of motion $T_{r}$ create higher flows and their existence
implies Liouville integrability of the system (\ref{ALH}). The Hamiltonian (%
\ref{clH}) creating the physical time flow\ (\ref{ALH}) is given by%
\begin{equation}
H=-T_{1}-T_{n-1}+2\limfunc{Tr}(\ln L(1))
\end{equation}%
Thus, we note that the physical time flow separates into three different
commuting flows generated by $T_{1}=\sum_{j=1}^{n}\psi _{j}\psi _{j+1}^{\ast
}$, $T_{n-1}=\sum_{j=1}^{n}\psi _{j}^{\ast }\psi _{j+1}$ and the determinant
(\ref{detL}); see the discussion in \cite{Suris97}. For instance, the flow
generated from $-T_{1}$ gives rise to the differential-difference equations 
\cite[Eqns (2.7) and (2.8)]{Suris97}%
\begin{equation}
\left\{ 
\begin{array}{c}
\partial _{t}\psi _{j}=\psi _{j-1}(1-\psi _{j}^{\ast }\psi _{j}) \\ 
\partial _{t}\psi _{j}^{\ast }=-\psi _{j+1}^{\ast }(1-\psi _{j}^{\ast }\psi
_{j})%
\end{array}%
\right. \;.  \label{right_flow}
\end{equation}%
The classical B\"{a}cklund transforms which we discuss in the next section
can be understood as time-discretisation of this particular flow.

\subsection{A one-variable family of B\"{a}cklund transforms}

In this article we are interested in one-variable families of B\"{a}cklund
transformations $\mathcal{B}^{\pm }(v):(\psi ,\psi ^{\ast })\mapsto (\tilde{%
\psi},\tilde{\psi}^{\ast })$ which correspond to the flows discussed by
Suris in \cite[Prop 1 and Prop 4]{Suris97} with $v$ playing the role of the
discrete time parameter $v\mathbb{Z}$. We summarise here some of the results
from \emph{loc. cit.} for the case of periodic boundary conditions to keep
this article self-contained and to introduce our notation. The novel aspect
in this article is that we relate these discrete time flows to the $Q^{\pm }$%
-operator which we introduce in a subsequent section.

The transform $\mathcal{B}^{+}(v)$ is defined implicitly in the standard
manner using a \emph{Darboux transformation}: one considers the following
local gauge transformation of the Lax operator (\ref{classicL}) 
\begin{equation}
D_{j+1}^{+}(u,v)L_{j}(u;\psi ,\psi ^{\ast })=L_{j}(u;\tilde{\psi},\tilde{\psi%
}^{\ast })D_{j}^{+}(u,v)\;,  \label{DL=LD}
\end{equation}%
where the so-called\emph{\ Darboux matrices }$D_{j}^{+}$\emph{\ }are assumed
to be of the following form, 
\begin{equation}
D_{j}^{+}(u,v)=\left( 
\begin{array}{cc}
v-ua_{j} & -ub_{j} \\ 
vc_{j} & -u%
\end{array}%
\right) \;.  \label{D}
\end{equation}%
Here $v$ is an additional variable related to the spectrality property of
the B\"{a}cklund transform $\mathcal{B}^{+}(v)$; see \cite{Sklyanin00} and
references therein for an explanation in the context of the Toda chain. N.B.
the dressing matrix is singular for $u=v$, $\det D_{j}^{+}(v,v)=0$, if we
require that $a_{j}=1-b_{j}c_{j}$ in close analogy with the case of the Toda
chain discussed in \emph{loc. cit.} The equality (\ref{DL=LD}) together with
the explicit form of the Darboux matrix (\ref{D}) then fixes the map $%
\mathcal{B}^{+}(v):(\psi ,\psi ^{\ast })\mapsto (\tilde{\psi},\tilde{\psi}%
^{\ast })$; compare with \cite[Eqn (3.1)]{Suris97}.

\begin{proposition}
The transformed variables $(\tilde{\psi}_{j},\tilde{\psi}_{j}^{\ast })$ obey
the following functional relations%
\begin{equation}
\left\{ 
\begin{array}{c}
\tilde{\psi}_{j}-\psi _{j}=v~\tilde{\psi}_{j-1}(1-\psi _{j}^{\ast }\tilde{%
\psi}_{j})\medskip \\ 
\tilde{\psi}_{j}^{\ast }-\psi _{j}^{\ast }=v~\psi _{j+1}^{\ast }(\psi
_{j}^{\ast }\tilde{\psi}_{j}-1)%
\end{array}%
\right. \;.  \label{clB}
\end{equation}%
In particular, the matrix elements in (\ref{D}) are $a_{j}=1+v\psi
_{j}^{\ast }\tilde{\psi}_{j-1}$, $b_{j}=-v\psi _{j}^{\ast }$, $c_{j}=\tilde{%
\psi}_{j-1}$.
\end{proposition}

\begin{proof}
A straightforward computation. Inserting (\ref{D}) into (\ref{DL=LD}) and
comparing coefficients of powers in the spectral variable $u$ yields the
asserted equalities.
\end{proof}

The identities (\ref{clB}) allow one to compute the transformed variables $(%
\tilde{\psi},\tilde{\psi}^{\ast })$ via recurrence upon expanding them in
power series with respect to the variable $v$. That is, if we set%
\begin{equation}
\tilde{\psi}_{j}(v)=\sum_{r\geq 0}v^{r}\tilde{\psi}_{j_{,}r}\qquad \text{\
and}\qquad \tilde{\psi}_{j}^{\ast }(v)=\sum_{r\geq 0}v^{r}\tilde{\psi}%
_{j_{,}r}^{\ast }  \label{psi_expansion}
\end{equation}%
then the above relations imply the recurrence identities%
\begin{equation}
\left\{ 
\begin{array}{l}
\tilde{\psi}_{j_{,}r}=\tilde{\psi}_{j-1,r-1}+\psi _{j}^{\ast
}\tsum\limits_{a+b=r-1}\tilde{\psi}_{j-1,a}\tilde{\psi}_{j,b}\medskip \\ 
\tilde{\psi}_{j,r}^{\ast }=\psi _{j}^{\ast }\psi _{j+1}^{\ast }\tilde{\psi}%
_{j,r-1}%
\end{array}%
\right. \;.  \label{recurrence}
\end{equation}%
Note that $\mathcal{B}^{+}(0)=\limfunc{Id}$ which fixes the initial
conditions $\tilde{\psi}_{j,0}=\psi _{j}$ and $\tilde{\psi}_{j,0}^{\ast
}=\psi _{j}^{\ast }$. By construction the map $\mathcal{B}^{+}(v)$ leaves
the Hamiltonians (\ref{TrL}) as well as (\ref{detL}) invariant, $T_{r}(\psi
,\psi ^{\ast })=T_{r}(\tilde{\psi},\tilde{\psi}^{\ast })$. Furthermore, one
has \cite[Prop 5]{Suris97}:

\begin{proposition}
$\mathcal{B}^{+}(v)$ preserves the Poisson structure (\ref{Poisson}).
\end{proposition}

Therefore the transform $\mathcal{B}^{+}(v)$ satisfies the defining
conditions of an integrable map \cite{Veselov91} and it then follows from a
general argument \cite{Sklyanin00} that the transforms $\mathcal{B}^{+}(x)$
and $\mathcal{B}^{+}(y)$ commute, 
\begin{equation}
\mathcal{B}^{+}(x)\circ \mathcal{B}^{+}(y)=\mathcal{B}^{+}(y)\circ \mathcal{B%
}^{+}(x)\;.  \label{Bcomm}
\end{equation}

One can also consider the inverse of the Darboux matrix (\ref{D}) and one
then arrives at a second transform $\mathcal{B}^{-}(v):(\psi ,\psi ^{\ast
})\mapsto (\hat{\psi},\hat{\psi}^{\ast })$ which in case of the infinite
chain is related to the inverse of $\mathcal{B}^{+}(-v)$ defined via the
relations \cite[Prop 4, Eqns (3.13) and (3.14)]{Suris97}%
\begin{equation}
\left\{ 
\begin{array}{c}
\hat{\psi}_{j}-\psi _{j}=-v\psi _{j-1}(1-\psi _{j}\hat{\psi}_{j}^{\ast
})\medskip \\ 
\hat{\psi}_{j}^{\ast }-\psi _{j}^{\ast }=v(1-\psi _{j}\hat{\psi}_{j}^{\ast })%
\hat{\psi}_{j+1}^{\ast }%
\end{array}%
\right. \;.  \label{clB-}
\end{equation}%
The proofs are analogous to the previous case. Both systems of equations, (%
\ref{clB}) and (\ref{clB-}), are approximations of the flow (\ref{right_flow}%
). Similarly, one can investigate the transforms for the flow generated by $%
-T_{n-1}$; see \cite[Prop 2 and Prop 3]{Suris97}. These turn out to be
related to the hermitian adjoints of the $Q^{\pm }$-operators, we therefore
restrict our discussion to (\ref{clB}) and (\ref{clB-}).

\section{Quantisation of the Ablowitz-Ladik system}

The following quantisation of the Ablowitz-Ladik model was first discussed
by Kulish \cite{Kulish81}, 
\begin{equation}
\left\{ 
\begin{array}{c}
\partial _{t}\beta _{j}=\beta _{j+1}-2\beta _{j}+\beta _{j-1}-\beta
_{j}^{\ast }\beta _{j}(\beta _{j+1}+\beta _{j-1})\medskip \\ 
\partial _{t}\beta _{j}^{\ast }=-\beta _{j+1}^{\ast }+2\beta _{j}^{\ast
}-\beta _{j-1}^{\ast }+\beta _{j}^{\ast }\beta _{j}(\beta _{j+1}^{\ast
}+\beta _{j-1}^{\ast })%
\end{array}%
\right. ~,  \label{qALH}
\end{equation}%
where $\{\beta _{j},\beta _{j}^{\ast }\}_{j\in \mathbb{Z}_{n}}$ are now
noncommutative variables satisfying the defining relation of the $q$-boson
or $q$-Heisenberg algebra $\mathcal{H}_{n}(q)$; c.f. \cite[Eqn (14)]%
{Kulish81} and \cite{Biedenharn89}. Throughout this article we assume $q$ to
be an indeterminate.

\begin{definition}
Let $\mathcal{H}_{n}(q)$ be the $\mathbb{C}(q)$-algebra generated by $%
\{\beta _{i},\beta _{i}^{\ast },q^{\pm N_{i}}\}_{i=1}^{n}$ subject to the
relations%
\begin{equation}
\left\{ 
\begin{array}{c}
q^{N_{i}}\beta _{j}=q^{-\delta _{ij}}\beta _{j}q^{N_{i}},\qquad
q^{N_{i}}\beta _{j}^{\ast }=q^{\delta _{ij}}\beta _{j}^{\ast
}q^{N_{i}}\medskip \\ 
\beta _{i}\beta _{j}^{\ast }-\beta _{j}^{\ast }\beta _{i}=\delta
_{ij}(1-q^{2})q^{2N_{i}},\quad \beta _{i}\beta _{i}^{\ast }-q^{2}\beta
_{i}^{\ast }\beta _{i}=1-q^{2}%
\end{array}%
\right. ~.  \label{qbosondef}
\end{equation}
\end{definition}

N.B. the notation $q^{N_{i}}$ is purely formal, it does \emph{not} mean that 
$q^{N_{i}}$ is given by exponentiating another generator $N_{i}$ although we
will often write $q^{a}(q^{N_{i}})^{m}$ as $q^{mN_{i}+a}$. We have also made
a small change to the usual definition of the $q$-boson algebra, see e.g. 
\cite{KlimykSchmuedgen97}, by multiplying one of the generators with an
extra factor, $\beta _{j}\rightarrow (1-q^{2})\beta _{j}$. This allows us to
derive from the second set of relations in (\ref{qbosondef}) the identity $%
\beta _{j}^{\ast }\beta _{j}=1-q^{2N_{j}}$ and, thus, we obtain the familiar
quantisation formula $[\cdot ,\cdot ]=-\imath \hbar \{\cdot ,\cdot
\}+O(\hbar ^{2})$ for the Poisson structure (\ref{Poisson}) when evaluating
the indeterminate as $q\rightarrow \exp (\imath \hbar \gamma /2)$,%
\begin{equation}
\lbrack \beta _{i},\beta _{j}^{\ast }]=\delta _{ij}(1-q^{2})(1-\beta
_{i}^{\ast }\beta _{i})=-\delta _{ij}~\imath \hbar \gamma (1-\beta
_{i}^{\ast }\beta _{i})+O(\hbar ^{2})\;,  \label{quantisation}
\end{equation}%
compare with \cite[Eqn (13)]{Kulish81}. Similar to the classical case, $%
\beta _{j}^{\ast }$ is initially an independent generator, but we may
consider representations of the $q$-boson algebra where $\beta _{j}^{\ast }$
is the hermitian adjoint of $\beta _{j}$; see the Fock representation (\ref%
{Fock}) below. This corresponds to the above mentioned reduction of the
Ablowitz-Ladik chain to the NLS model and due to the change to imaginary
time, $t\rightarrow \imath t$, one then must choose $\gamma =\imath c$ with $%
c\in \mathbb{R}$ being the coupling constant in the QNLS Hamiltonian.

The form of the Hamiltonian closely resembles the classical one,%
\begin{equation}
H=-\sum_{j=1}^{n}(\beta _{j}\beta _{j+1}^{\ast }+\beta _{j}^{\ast }\beta
_{j+1}-2(1-q^{2})N_{j}),  \label{qLH}
\end{equation}%
where the $N_{j}$'s are a set of additional generators not contained in the
original $q$-boson algebra obeying%
\begin{equation}
N_{i}q^{N_{j}}=q^{N_{j}}N_{i},\quad \beta _{i}(N_{j}-\delta
_{ij})=N_{j}\beta _{i},\quad \beta _{i}^{\ast }(N_{j}+\delta
_{ij})=N_{j}\beta _{i}^{\ast }\;.  \label{number_op}
\end{equation}%
If we identify $\beta _{i}^{\ast }$ with the creation and $\beta _{i}$ with
the annihilation of a $q$-boson at site $i$, then we easily recognise that $%
N_{i}$ plays the role of a number operator on this site. We shall denote the
total particle number operator by $N=\sum_{j=1}^{n}N_{j}$.

In close analogy to the classical case (\ref{classicL}) one defines the
monodromy matrix of the quantised system as a product of the local Lax
matrices%
\begin{equation}
L(u)=L_{n}(u)\cdots L_{1}(u)=\left( 
\begin{array}{cc}
A(u) & B(u) \\ 
C(u) & D(u)%
\end{array}%
\right) ,\qquad L_{j}(u)=\left( 
\begin{array}{cc}
1 & u\beta _{j}^{\ast } \\ 
\beta _{j} & u%
\end{array}%
\right)   \label{mom}
\end{equation}%
and one then obtains the quantum versions of the spectral invariants (\ref%
{TrL}) of the classical system,%
\begin{equation}
T(u)=\limfunc{Tr}L(u)=A(u)+z~D(u)=\sum_{r=0}^{n}u^{r}T_{r}~,  \label{T}
\end{equation}%
where $z$ is an additional quasi-periodicity parameter which we will need
later on. The explicit form of the quantum integrals of motion $T_{r}$ has
been derived in \cite[Prop 3.10, Eqn (3.43)]{CMP13}: set $a_{i}=\beta
_{i}\beta _{i+1}^{\ast }$ for $i=1,\ldots ,n-1$ and $a_{n}=z\beta _{n}\beta
_{1}^{\ast }$ (right or clockwise hopping) then%
\begin{equation}
T_{r}=\sum_{w=i_{1}\cdots i_{r}}\frac{[a_{i_{1}},[a_{i_{2}},\ldots \lbrack
a_{i_{r-1}},a_{i_{r}}]_{q^{2}}\ldots ]_{q^{2}}}{(1-q^{2})^{r-1}}~
\label{T_r}
\end{equation}%
with $[X,Y]_{q^{2}}=XY-q^{2}YX$ and the sum is running over all cyclically
ordered words $w$ with letters $1\leq i_{j}\leq n$ occurring at most once.
Moreover, $T_{0}=1$, $T_{n}=z$ and we set $T_{r}=0$ for all $r>n$ and $r<0$.
Note that $T_{1}=\sum_{j=1}^{n}\beta _{j}\beta _{j+1}^{\ast }$ and that $%
T_{n}$ simplifies to $T_{n}=\sum_{j=1}^{n}\beta _{j}^{\ast }\beta _{j+1}$.
Thus, as in the classical case one has a splitting of the quantum
Hamiltonian (\ref{qLH}) into `right movers' and `left movers'.

One now proves quantum integrability of the system (\ref{qALH}) via the
following solution of the Yang-Baxter equation \cite{Kulish81}, 
\begin{equation}
R_{12}(u/v)L_{13}(u)L_{23}(v)=L_{23}(v)L_{13}(u)R_{12}(u/v),  \label{RLL}
\end{equation}%
where the $R$-matrix in $\limfunc{End}(\mathbb{C}^{2}\otimes \mathbb{C}%
^{2})\cong \limfunc{End}\mathbb{C}^{4}$ with the conventions used in the
definition (\ref{qLH}) of $L$ reads%
\begin{equation}
R(u)=\left( 
\begin{array}{cccc}
\frac{uq^{2}-1}{u-1} & 0 & 0 & 0 \\ 
0 & q^{2} & \frac{q^{2}-1}{u-1}u & 0 \\ 
0 & \frac{q^{2}-1}{u-1} & 1 & 0 \\ 
0 & 0 & 0 & \frac{uq^{2}-1}{u-1}%
\end{array}%
\right) \;.  \label{R}
\end{equation}%
By the standard argument one obtains as an immediate consequence the
following proposition.

\begin{proposition}
\label{prop:qint}The subalgebra $\mathcal{A}_{n}\subset \mathcal{H}_{n}$
generated by the $T_{r}$ is commutative, $[T_{r},T_{r^{\prime }}]=0$.
\end{proposition}

\begin{proof}
One verifies that (\ref{R}) is invertible. Using the cyclicity of the trace
in (\ref{T}) one obtains $[T(u),T(v)]=0$. This proves the assertion for $z=1$%
. The case of general $z$ follows by the same argument using that $R(u/v)$
commutes with $\left( 
\begin{smallmatrix}
1 & 0 \\ 
0 & z%
\end{smallmatrix}%
\right) \otimes \left( 
\begin{smallmatrix}
1 & 0 \\ 
0 & z%
\end{smallmatrix}%
\right) $.
\end{proof}

We now consider a particular representation of the $q$-boson algebra which
allows one to apply the quantum inverse scattering method or algebraic Bethe
ansatz to construct a common eigenbasis.

\subsection{The Fock representation of q-bosons}

For the sake of completeness we recall the following representation of the $%
q $-boson algebra in terms of $q$-difference operators, which can be found
in \cite{KlimykSchmuedgen97}. Consider the ring of polynomials $\mathfrak{R}%
^{n}=\mathbb{C}(q)[\xi _{1},\ldots ,\xi _{n}]$ with rational coefficients in
the indeterminate $q$, where the $\xi _{i}$'s are some auxiliary variables.
Let $\tau _{i}$ be the shift operator with respect to $\xi _{i}$,%
\begin{equation}
(\tau _{i}^{\pm 1}f)(\xi )=f(\xi _{1},...,\xi _{i}q^{\pm 1},...,\xi
_{n}),\qquad f\in \mathfrak{R}^{n}  \label{qshift}
\end{equation}%
and define $\mathcal{D}_{i}$ to be the following $q$-derivative with respect
to $\xi _{i}$,%
\begin{equation}
(\mathcal{D}_{i}f)(\xi )=\frac{f(\xi )-f(\xi _{1},...,q^{2}\xi _{i},...,\xi
_{n})}{\xi _{i}-q^{2}\xi _{i}}~.  \label{qderivative}
\end{equation}%
In the limit $q\rightarrow 1$ one recovers the ordinary partial derivative $%
\partial _{i}$. We introduce a hermitian bilinear form $\mathfrak{R}%
^{n}\times \mathfrak{R}^{n}\rightarrow \mathbb{C}(q)$ by setting 
\begin{equation}
\langle f,g\rangle =\bar{f}(\mathcal{D}_{1},...,\mathcal{D}_{n})g(\xi
_{1},...,\xi _{n})|_{\xi _{1}=\cdots =\xi _{n}=0}  \label{inner_product}
\end{equation}%
and define $\mathcal{F}^{n}=\overline{\mathfrak{R}^{n}}$ to be the
completion with respect to this inner product. We will simply write $%
\mathcal{F}$ for $\mathcal{F}^{1}$ and we have that $\mathcal{F}^{n}=%
\mathcal{F}^{\otimes n}$. The following result is contained in \cite%
{KlimykSchmuedgen97}.

\begin{proposition}[Fock representation]
The map $\mathcal{H}_{n}\rightarrow \limfunc{End}\mathcal{F}^{n}$ defined via%
\begin{equation}
\beta _{i}\mapsto (1-q^{2})\mathcal{D}_{i},\qquad \beta _{i}^{\ast }\mapsto 
\hat{\xi}_{i},\qquad q^{\pm N_{i}}\mapsto \tau _{i}^{\pm 1}~,  \label{Fock}
\end{equation}%
where $\hat{\xi}_{i}$ is the multiplication operator with $\xi _{i}$,
defines a simple module for $\mathcal{H}_{n}$ for all $n\geq 1$.
\end{proposition}

There is an alternative formulation of the same module used in \cite[Prop 3.3%
]{CMP13}: let $\mathcal{I}\subset \mathcal{H}_{n}$ be the \emph{left} ideal
generated by the elements $\beta _{i}$ and $(1-q^{\pm N_{i}})$. Define a
highest weight vector $|0\rangle =1+\mathcal{I}$, the \textquotedblleft
pseudo-vacuum\textquotedblright , in the quotient $\mathcal{H}_{n}/\mathcal{I%
}$ with $1$ the identity element. For any partition $\lambda $ with at most $%
n$ parts set%
\begin{equation}
|\lambda \rangle =\prod_{j=1}^{n}\frac{(\beta _{j}^{\ast })^{m_{j}(\lambda )}%
}{(q^{2})_{m_{j}(\lambda )}}~|0\rangle ,\qquad
(q^{2})_{m}:=\prod_{j=1}^{m}(1-q^{2j}),  \label{can_basis}
\end{equation}%
where $m_{j}(\lambda )$ is the multiplicity of columns of height $j$ in the
Young diagram of $\lambda $. It is easy to verify that the map $|\lambda
\rangle \mapsto \prod_{j=1}^{n}\xi _{j}^{m_{j}(\lambda
)}/(q^{2})_{m_{j}(\lambda )}$ provides a module isomorphism $\mathcal{H}_{n}/%
\mathcal{I~}\widetilde{\mathcal{\rightarrow }}~\mathcal{F}^{n}$ and we shall
henceforth identify both modules. We denote by $\langle \lambda |$ the dual
basis of (\ref{can_basis}) and the map $\mathcal{\tilde{F}}^{n}\rightarrow 
\mathcal{F}^{n}$ given by $\langle \lambda |\mapsto b_{\lambda
}(q^{2})|\lambda \rangle $ with $\mathcal{\tilde{F}}^{n}$ denoting the dual
space and $b_{\lambda }=\prod_{j\geq 1}(q^{2})_{m_{j}(\lambda )}$, then
introduces an inner product on $\mathcal{F}^{n}$. Note that if $q$ is
evaluated at a root of unity the module ceases to be simple.

\begin{remark}
\textrm{The eigenspaces of the particle number operator $N=\sum_{j=1}^{n}%
\hat{\xi}_{j}\partial _{\xi _{j}}$ are the subspaces $\mathcal{F}%
_{k}^{n}\subset \mathcal{F}^{n}$ spanned by the vectors $\{|\lambda \rangle
\}_{\lambda _{1}=k}$ for $k\in \mathbb{Z}_{\geq 0}$. In \cite{CMP13} it has
been shown that $\mathcal{F}_{k}^{n}$ carries an $U_{q}(\widehat{sl}_{n})$%
-action and can be identified with the Kirillov-Reshetikhin module of
highest weight $k\omega _{1}$ with $\omega _{1}$ being the first fundamental
weight. The basis (\ref{can_basis}) is then Lusztig's canonical or
Kashiwara's global crystal basis for this module. }
\end{remark}

\subsection{Completeness of the Bethe ansatz}

The existence of the highest weight representation (\ref{Fock}) allows for
the application of the quantum inverse scattering method \cite{FST79} or
algebraic Bethe ansatz and one obtains the following important result.

Denote by $\Bbbk =\mathbb{C}\{\!\{q\}\!\}$ the algebraically closed field of
Puiseux series in the indeterminate $q$.

\begin{theorem}
In the Fock representation the difference operators given by (\ref{T_r})
possess a common eigenbasis $\{|y_{\lambda }\rangle \}_{\lambda }$ where 
\begin{equation}
|y_{\lambda }\rangle =B(y_{1}^{-1})\cdots B(y_{\lambda _{1}}^{-1})|0\rangle
\label{Bethe_vector}
\end{equation}%
and $\lambda $ ranges over the partitions with at most $n$ parts. The
partition $\lambda $ labels the solutions $y_{\lambda }=(y_{1},\ldots
,y_{\lambda _{1}})\in \Bbbk ^{\lambda _{1}}$ of the following coupled set of
equations%
\begin{equation}
y_{i}^{n}=z\prod_{j\neq i}\frac{y_{i}q^{2}-y_{j}}{y_{i}-y_{j}q^{2}},\qquad
i=1,2,\ldots ,\lambda _{1}\;.  \label{BAE}
\end{equation}%
The eigenvalue equation reads,%
\begin{equation}
T(u)|y_{\lambda }\rangle =\left( \tprod_{j=1}^{\lambda _{1}}\frac{%
1-uq^{2}y_{j}}{1-u~y_{j}}+zu^{n}\tprod_{j=1}^{\lambda _{1}}\frac{q^{2}-uy_{j}%
}{1-u~y_{j}}\right) |y_{\lambda }\rangle \;.  \label{specT}
\end{equation}
\end{theorem}

\begin{remark}
\textrm{The quantum inverse scattering method for the $q$-boson model was
discussed in \cite{Kulish81}, \cite{BBP93} and \cite{BIK98} where the Bethe
ansatz equations (\ref{BAE}) can be found. A discussion of the coordinate
Bethe ansatz for the discrete QNLS model arriving at the same set of
equations appeared also in \cite{vanDiejen06}, where a proof of the
completeness of the Bethe ansatz solutions can be found for }$q=\varepsilon $%
\textrm{\ with $-1<\varepsilon <1$. Completeness of the Bethe ansatz for $q$
an indeterminate together with a coordinate ring description of their
solutions as quotient of the spherical Hecke algebra was obtained in \cite[%
Section 7]{CMP13}. In order to find solutions of (\ref{BAE}) one then needs
to work over the algebraically closed field of Puiseux series in $q$. }
\end{remark}

\section{Two $Q$-operators for the $q$-boson model}

In order to construct the quantum analogue of the B\"{a}cklund transform (%
\ref{clB}) and its inverse one considers solutions $D^{\pm }$, $L^{\pm }$ to
the Yang-Baxter equation%
\begin{equation}
D_{12}^{\pm }(u,v)L_{13}(u)L_{23}^{\pm }(v)=L_{23}^{\pm
}(v)L_{13}(u)D_{12}^{\pm }(u,v),  \label{DLL=LLD}
\end{equation}%
where $L(u)$ is the local quantum Lax matrix in (\ref{mom}) and $D^{\pm
}(u,v),L^{\pm }(v)$ are respectively $2\times \frac{\infty }{2}$ and $\frac{%
\infty }{2}\times \frac{\infty }{2}$ matrices with entries in the $q$-boson
algebra. As explained in \cite{Sklyanin00} for the Toda chain, the
Yang-Baxter equation (\ref{DLL=LLD}) should be seen as a quantum analogue of
the classical relation (\ref{DL=LD}) for the Darboux matrix: due to the
noncommutative matrix elements in the quantum case an extra matrix $L^{\pm }$
is required. Thus, in the quantum system one faces the problem of finding 
\emph{two} matrices, $D^{\pm }$ and $L^{\pm }$, to construct the quantum
analogue of the B\"{a}cklund transform (\ref{clB}).

Define the following half-infinite matrices with entries in the $q$-boson
algebra $\mathcal{H}_{n}$ (c.f. \cite[Eqn (3.20)]{CMP13}), 
\begin{equation}
L_{i}^{+}(v)=\left( (-v)^{m}\frac{(\beta _{i}^{\ast })^{m}\beta
_{i}^{m^{\prime }}}{(q^{2})_{m}}\right) _{m,m^{\prime }\geq 0}  \label{L+}
\end{equation}%
$\quad $and\quad 
\begin{equation}
L_{i}^{-}(v)=\left( v^{m}q^{m(m+1)}\frac{\beta _{i}^{m^{\prime }}(\beta
_{i}^{\ast })^{m}}{(q^{2})_{m}}\right) _{m,m^{\prime }\geq 0}\;.  \label{L-}
\end{equation}%
The operators $L^{\pm }$ which are not present in the classical equation
enter in the definition of the $Q^{\pm }$-operators 
\begin{equation}
Q^{\pm }(v)=\limfunc{Tr}z^{N}L_{n}^{\pm }(v)\cdots L_{2}^{\pm }(v)L_{1}^{\pm
}(v)=\sum_{r\geq 0}Q_{r}^{\pm }v^{r},  \label{Qpm}
\end{equation}%
where $z^{N}=(z^{m~\delta _{mm^{\prime }}})$ and the trace is formally
defined with $Q^{\pm }(v)\in \mathbb{C}[\![v]\!]\otimes \mathcal{H}_{n}$.
That is, $Q^{\pm }(v)$ should be thought of as current operators: as formal
power series in the variable $v$ with the coefficients $Q_{r}^{\pm }$ being
elements in the $q$-boson algebra $\mathcal{H}_{n}$.

\begin{lemma}
The coefficients in the expansions (\ref{Qpm}) are%
\begin{eqnarray}
Q_{r}^{+} &=&(-1)^{r}\sum_{\alpha \vdash r}z^{\alpha _{n}}\frac{(\beta
_{1}^{\ast })^{\alpha _{n}}(\beta _{1}\beta _{2}^{\ast })^{\alpha
_{1}}\cdots (\beta _{n-1}\beta _{n}^{\ast })^{\alpha _{n-1}}\beta
_{n}^{\alpha _{n}}}{(q^{2})_{\alpha _{1}}\cdots (q^{2})_{\alpha
_{n-1}}(q^{2})_{\alpha _{n}}}  \label{Q+r} \\
Q_{r}^{-} &=&\sum_{\alpha \vdash r}z^{\alpha _{n}}\frac{\beta _{n}^{\alpha
_{n}}(\beta _{n-1}\beta _{n}^{\ast })^{\alpha _{n-1}}\cdots (\beta _{1}\beta
_{2}^{\ast })^{\alpha _{1}}(\beta _{1}^{\ast })^{\alpha _{n}}}{%
(q^{2})_{\alpha _{1}}\cdots (q^{2})_{\alpha _{n-1}}(q^{2})_{\alpha _{n}}}%
\prod_{i=1}^{n}q^{\alpha _{i}(\alpha _{i}+1)}\;,  \label{Q-r}
\end{eqnarray}%
where the sums range over all compositions $\alpha =(\alpha _{1},\ldots
,\alpha _{n})\in \mathbb{Z}_{\geq 0}^{n}$ of $r\geq 0$.
\end{lemma}

\begin{proof}
Via induction in $n$ one verifies that the monodromy matrices $T^{\pm
}(v)=L_{n}^{\pm }(v)\cdots L_{2}^{\pm }(v)L_{1}^{\pm }(v)$ are given by
(c.f. \cite[Eqn (3.48)]{CMP13})%
\begin{equation*}
T_{m^{\prime }m}^{+}(v)=z^{m}\sum_{\alpha }\frac{(-1)^{m+|\alpha
|}v^{m+|\alpha |}}{(q^{2})_{m}}\frac{(\beta _{1}^{\ast })^{m}(\beta
_{1}\beta _{2}^{\ast })^{\alpha _{1}}\cdots (\beta _{n-1}\beta _{n}^{\ast
})^{\alpha _{n-1}}\beta _{n}^{m^{\prime }}}{(q^{2})_{\alpha _{1}}\cdots
(q^{2})_{\alpha _{n-1}}}
\end{equation*}%
and%
\begin{multline*}
T_{m^{\prime }m}^{-}(v)= \\
z^{m}\sum_{\alpha }\frac{v^{m+|\alpha |}q^{m(m+1)}}{(q^{2})_{m}}\frac{\beta
_{n}^{\alpha _{n}}(\beta _{n-1}\beta _{n}^{\ast })^{\alpha _{n-1}}\cdots
(\beta _{1}\beta _{2}^{\ast })^{\alpha _{1}}(\beta _{1}^{\ast })^{\alpha
_{n}}}{(q^{2})_{\alpha _{1}}\cdots (q^{2})_{\alpha _{n-1}}(q^{2})_{\alpha
_{n}}}\prod_{i=1}^{n-1}q^{\alpha _{i}(\alpha _{i}+1)}
\end{multline*}%
where the sums range over all compositions $\alpha =(\alpha _{1},\ldots
,\alpha _{n-1})\in \mathbb{Z}_{\geq 0}^{n-1}$. Summing over the diagonal
matrix elements with $m=m^{\prime }$ and fixing the degree of $v$ the
assertion now follows.
\end{proof}


\subsection{Functional relations in the $q$-boson algebra}

The characteristic property which prompts us to identify (\ref{Qpm}) as the
analogue of Baxter's $Q$-operator for the $q$-boson model is the following
set of functional relations. The first was already proved in \cite[Prop 3.12]%
{CMP13}, we repeat it here for completeness, the second identity (\ref{TQ-})
is new.

\begin{theorem}
We have the following identities in the $q$-boson algebra%
\begin{eqnarray}
T(u)Q^{+}(u) &=&Q^{+}(uq^{2})+\Delta (u)Q^{+}(uq^{-2}),  \label{TQ+} \\
Q^{-}(u)T(u) &=&Q^{-}(uq^{-2})+\Delta (uq^{2})Q^{-}(uq^{2})~,  \label{TQ-}
\end{eqnarray}%
where the coefficient is given by $\Delta (u)=zq^{2N}u^{n}$.
\end{theorem}

\begin{proof}
For the proof of the first equation we refer the reader to \cite[Prop 3.12]%
{CMP13}. The proof of the second equation follows along similar lines.
Namely, we consider the action of $L_{13}^{-}(u)L_{23}(u)$ on the tensor
product $V\otimes \mathbb{C}^{2}\otimes \mathcal{H}_{1}$ with $%
V=\tbigoplus_{m\geq 0}\mathbb{C}v_{m}$, where we identify the basis vectors
in $\mathbb{C}^{2}$ as $v_{0},v_{1}$. Let $X$ be any element in $\mathcal{H}%
_{1}$. Then a straightforward computation yields%
\begin{multline*}
L_{13}^{-}(u)L_{23}(u)v_{m}\otimes v_{0}\otimes X= \\
\frac{u^{m}q^{m(m+1)}}{(q^{2})_{m}}\sum_{m^{\prime }\geq 0}v_{m^{\prime
}}\otimes v_{0}\otimes \beta ^{m^{\prime }}(\beta ^{\ast
})^{m}X+v_{m^{\prime }}\otimes v_{1}\otimes \beta ^{m^{\prime }}(\beta
^{\ast })^{m}\beta X
\end{multline*}%
and for $m>0$,%
\begin{multline*}
L_{13}^{-}(u)L_{23}(u)v_{m-1}\otimes v_{1}\otimes X= \\
\frac{u^{m}q^{m(m-1)}}{(q^{2})_{m-1}}\sum_{m^{\prime }\geq 0}v_{m^{\prime
}}\otimes v_{1}\otimes \beta ^{m^{\prime }}(\beta ^{\ast
})^{m-1}f+v_{m^{\prime }}\otimes v_{0}\otimes \beta ^{m^{\prime }}(\beta
^{\ast })^{m}X
\end{multline*}%
Consider now the subspace $W\subset V\otimes \mathbb{C}^{2}$ spanned by the
vectors $\{w_{m}=v_{m}\otimes v_{0}+v_{m-1}\otimes v_{0}~|~m>0\}$ and $%
w_{0}=v_{0}\otimes v_{0}$. Employing the commutation relation%
\begin{equation*}
(\beta ^{\ast })^{m}\beta =(\beta ^{\ast })^{m-1}(1-q^{-2m})+q^{-2m}\beta
(\beta ^{\ast })^{m},
\end{equation*}%
which is easily verified by induction, one then finds from the two
identities above that on $W\otimes \mathcal{H}_{1}$ we have%
\begin{equation*}
L_{13}^{-}(u)L_{23}(u)w_{m}\otimes X=\frac{u^{m}q^{-2m}q^{m(m+1)}}{%
(q^{2})_{m-1}}\sum_{m^{\prime }\geq 0}w_{m^{\prime }}\otimes \beta
^{m^{\prime }}(\beta ^{\ast })^{m}X\;.
\end{equation*}%
Next we consider the action on the quotient space $W^{\prime }=V\otimes 
\mathbb{C}^{2}/W$. A basis is given by the vectors $\{v_{m}\otimes
v_{1}~|~m\geq 1\}$ and one arrives at%
\begin{multline*}
L_{13}^{-}(u)L_{23}(u)v_{m}\otimes v_{1}\otimes X= \\
\frac{u^{m+1}q^{m(m+1)}}{(q^{2})_{m}}\sum_{m^{\prime }\geq 0}v_{m^{\prime
}}\otimes v_{1}\otimes (\beta ^{m^{\prime }}\beta ^{\ast ~m}-\beta
^{m^{\prime }+1}\beta ^{\ast ~m+1})X+\cdots ,
\end{multline*}%
where the omitted terms lie in $W\otimes \mathcal{H}_{1}$. Since 
\begin{equation*}
\beta ^{m^{\prime }}\beta ^{\ast ~m}-\beta ^{m^{\prime }+1}\beta ^{\ast
~m+1}=q^{2m+2}\beta ^{m^{\prime }}\beta ^{\ast ~m}q^{2N}
\end{equation*}%
we deduce that the product of $L$-operators block-decomposes as follows,%
\begin{equation*}
L_{13}^{-}(u)L_{23}(u)=\left( 
\begin{array}{cc}
L^{-}(uq^{-2}) & \ast \\ 
0 & uq^{2N+2}L^{-}(uq^{2})%
\end{array}%
\right)
\end{equation*}%
and the assertion (\ref{TQ-}) is now easily obtained by taking the trace of
the product 
\begin{equation*}
T_{0}^{-}(u)T_{0^{\prime }}(u)=L_{0n}^{-}(u)L_{0^{\prime }n}(u)\cdots
L_{02}^{-}(u)L_{0^{\prime }2}(u)L_{01}^{-}(u)L_{0^{\prime }1}(u)
\end{equation*}%
of the monodromy matrices and applying the above block decomposition to each
factor $L_{0j}^{-}(u)L_{0^{\prime }j}(u)$.
\end{proof}

In addition to the functional equations (\ref{TQ+}) and (\ref{TQ-}) one
needs to prove that 
\begin{equation}
\lbrack T(u),Q^{\pm }(v)]=[Q^{\pm }(u),Q^{\pm }(v)]=[Q^{\pm }(u),Q^{\mp
}(v)]=0  \label{all_commute}
\end{equation}%
for arbitrary values of the variables $u,v$. The first and second relation
for the $Q^{+}$ operator has been proved in \cite[Props 3.7 and 3.8]{CMP13}
by constructing explicit solutions to the Yang-Baxter equation (\ref{DLL=LLD}%
) with \cite[Eqn (3.23)]{CMP13} 
\begin{equation}
D^{+}(u,v)=\left( 
\begin{array}{cc}
1-\frac{u}{v}q^{2N} & -\frac{u}{v}\beta ^{\ast } \\ 
\beta & -u/v%
\end{array}%
\right)  \label{D+}
\end{equation}%
as well as finding an additional solution of the equation $%
R_{12}^{+}(u,v)L_{13}^{+}(u)L_{23}^{+}(v)=L_{23}^{+}(v)L_{13}^{+}(u)R_{12}^{+}(u,v) 
$. The above expression for $D^{+}$ yields the desired $2\times \frac{\infty 
}{2}$ matrix in (\ref{DLL=LLD}) if the $q$-boson algebra elements in $D^{+}$
are evaluated in the Fock space representation (\ref{Fock}). Note the close
resemblance of the quantum Darboux matrix (\ref{D+}) with the classical one (%
\ref{D}). In principle one can proceed in the same manner for $Q^{-}$ and
one then finds%
\begin{equation}
D^{-}(u,v)=\left( 
\begin{array}{cc}
q^{2N} & \beta ^{\ast } \\ 
\frac{v}{u}\beta & 1-\frac{v}{u}q^{2N+2}%
\end{array}%
\right)  \label{D-}
\end{equation}%
which gives a second quantum Darboux matrix which closely resembles the
inverse of (\ref{D+}) \cite[Eqn (3.25)]{CMP13} in accordance with the
classical case. To establish (\ref{all_commute}) one needs to find yet two
other solutions of the Yang-Baxter equations $%
R_{12}^{-}(u,v)L_{13}^{-}(u)L_{23}^{-}(v)=L_{23}^{-}(v)L_{13}^{-}(u)R_{12}^{-}(u,v) 
$ and $R_{12}^{\prime
}(u,v)L_{13}^{+}(u)L_{23}^{-}(v)=L_{23}^{-}(v)L_{13}^{+}(u)R_{12}^{\prime
}(u,v)$. Here we shall instead take advantage of the already established
functional relations (\ref{TQ+}), (\ref{TQ-}) as well as Prop \ref{prop:qint}
to give a much shorter and less computational proof of (\ref{all_commute}).

\begin{corollary}
The operator coefficients of the formal power series $T(u),Q^{\pm }(u)$ all
commute, i.e. we have $[T_{r},Q_{s}^{\pm }]=[Q_{r}^{+},Q_{s}^{-}]=0$ for all 
$r,s\geq 0$.
\end{corollary}

\begin{proof}
Rewriting the functional equations (\ref{TQ+}) and (\ref{TQ-}) in terms of
coefficients we find that%
\begin{eqnarray}
(1-q^{2r})Q_{r}^{+}
&=&%
\sum_{s=1}^{r}(-1)^{s-1}T_{s}Q_{r-s}^{+}+(-1)^{n}zq^{2(N+r-n)}Q_{r-n}^{+}~, 
\notag \\
(q^{-2r}-1)Q_{r}^{-}
&=&\sum_{s=1}^{r}Q_{r-s}^{-}T_{s}+zq^{2(N+r)}Q_{r-n}^{-},  \label{Q2T}
\end{eqnarray}%
where we set $Q_{r}^{\pm }=0$ for $r<0$. Noting that $Q_{0}^{\pm }=1$ by
definition and that for each $r>0$ the terms on the right hand side of both
identities only involve $Q_{s}^{\pm }$ with $s<r$ we find that the $%
Q_{r}^{\pm }$ are polynomials in the $T_{r}$'s. But since the latter commute
among themselves, $[T_{r},T_{s}]=0$, the assertions now easily follow.
\end{proof}

\begin{corollary}
\label{cor:qwronskian}The transfer matrix $T(u)$ and\ $Q^{\pm }(u)$ obey the
additional relations%
\begin{equation}
1=Q^{+}(u)Q^{-}(uq^{-2})-zu^{n}q^{2N}Q^{+}(uq^{-2})Q^{-}(u)
\label{qWronskian}
\end{equation}%
and%
\begin{equation}
T(u)=\left\vert 
\begin{array}{cc}
Q^{+}(uq^{2}) & \Delta (u)Q^{+}(uq^{-2}) \\ 
\Delta (uq^{2})Q^{-}(uq^{2}) & Q^{-}(uq^{-2})%
\end{array}%
\right\vert \;.  \label{T=detQ}
\end{equation}%
As explained above these equalities should be understood as identities in
terms of $T_{r},Q_{r}^{\pm }\in \mathcal{H}_{n}$.
\end{corollary}

\begin{proof}
Consider the triple product $T(u)Q^{+}(u)Q^{-}(u)$. From (\ref{TQ+}) and (%
\ref{TQ-}) one deduces that the formal power series $W(u)=\sum_{r\geq
0}u^{r}W_{r}=Q^{+}(u)Q^{-}(uq^{-2})-zu^{n}q^{2N}Q^{+}(uq^{-2})Q^{-}(u)$
obeys $W(u)=W(uq^2)$ and, hence, $W_{r}=0$ for $r>0$ since $q$ is an
arbitrary indeterminate. Setting $u=0$ we find $W(0)=1$ and the first
assertion follows. The second equality (\ref{T=detQ}) is now easily deduced
from (\ref{qWronskian}) by using once more (\ref{TQ+}) and (\ref{TQ-}).
\end{proof}

\subsection{Specialisation to the Fock space representation}

It is worth emphasising that so far we have worked on the level of the $q$%
-boson algebra, that is, the functional relations derived hold regardless of
the representation chosen. We now specialise to the Fock space
representation (\ref{Fock}).

In the Fock representation the coefficients (\ref{Q+r}), (\ref{Q-r}) of the $%
Q$-operators become the following $q$-difference operators,%
\begin{eqnarray}
Q_{r}^{+} &=&(-1)^{r}(1-q^{2})^{r}\sum_{\alpha \vdash r}z^{\alpha _{n}}\frac{%
\hat{\xi}^{\alpha \prime {}}D^{\alpha }}{(q^{2})_{\alpha }}\;, \\
Q_{r}^{-} &=&(1-q^{2})^{r}\sum_{\alpha \vdash r}z^{\alpha _{n}}\frac{%
D^{\alpha }\hat{\xi}^{\alpha \prime {}}}{(q^{2})_{\alpha }}%
\prod_{i=1}^{n}q^{\alpha _{i}(\alpha _{i}+1)}\;,
\end{eqnarray}%
%
%
where $\alpha \prime {}=(\alpha _{n},\alpha _{1},\alpha _{2},\ldots ,\alpha
_{n-1})$. We can now take matrix elements of the $Q^{\pm }$-operators and
the $L^{\pm }$-operators in the Fock space representation. For the local Lax
matrix (\ref{mom}) we obtain for the vertex configurations in Figure \ref%
{fig:Tvertexmodel} the Boltzmann weights \cite[Fig 3]{CMP13}%
\begin{equation}
\langle 0,m|L(u)|0,m\rangle =\langle 1,m-1|L(u)|0,m\rangle =1
\end{equation}%
and%
\begin{equation}
\langle 0,m+1|L(u)|1,m\rangle =u(1-q^{2m+2}),\qquad \langle
1,m|L(u)|1,m\rangle =u\;.
\end{equation}%
From the $L^{\pm }$-operators we find the following Boltzmann weights for
the vertex configurations in Figure \ref{fig:Qvertexmodels},%
\begin{equation}
\langle c,d|L^{+}(u)|a,b\rangle =(-u)^{a}\QATOPD[ ] {d}{a}_{q^{2}}
\end{equation}%
and%
\begin{equation}
\langle c,d|L^{-}(u)|a,b\rangle =u^{a}q^{a(a+1)}\QATOPD[ ] {a+b}{b}_{q^{2}}\;
\end{equation}%
with $\QATOPD[ ] {m}{n}_{q^{2}}=\frac{(q^{2})_{m}}{(q^{2})_{n}(q^{2})_{m-n}}$
for $m>n$ and zero otherwise. The matrix elements of the $Q^{\pm }$%
-operators then yield the weighted sums over the vertex configurations of a
single lattice row, the row-to-row transfer matrices. According to (\ref%
{all_commute}) these define exactly solvable lattice models. As we have
interpreted $Q^{\pm }(u)$ so far as formal power series with coefficients in
the $q$-boson algebra we briefly explain how these give rise to well-defined
operators in Fock space.

For fixed particle number $k$ the formal power series $Q^{+}(u)|_{\mathcal{F}%
_{k}^{n}}$ is a well-defined operator as $Q_{r}^{+}|_{\mathcal{F}_{k}^{n}}=0$
for $r>k$. In fact, from (\ref{Q+r}) one deduces that $Q_{k}^{+}|_{\mathcal{F%
}_{k}^{n}}$ is the discrete translation operator: the sum in (\ref{Q+r}) for 
$r=k$ is the sum over all possible $k$-particle configurations shifting each
particle by one site forward. If $r>k$ the corresponding operator $%
Q_{r}^{+}|_{\mathcal{F}_{k}^{n}}$ would shift more particles forward than
are in the system, hence its matrix elements vanish.

In contrast, the coefficients of $Q^{-}(u)|_{\mathcal{F}_{k}^{n}}$ are in
general all nonzero on the $k$-particle space, since according to (\ref{L-})
and (\ref{TQ-}) particles are first created at the neighbouring site before
they are annihilated at their place of origin. Thus, under the action of $%
Q^{-}(u)|_{\mathcal{F}_{k}^{n}}$ a single particle can travel several times
around the lattice picking up a factor $z$ each time it completes a
round-trip. We therefore interpret the power of the quasi-periodicity
variable $z$ in (\ref{T}) as a winding number. For any finite winding number 
$p\geq 0$ the corresponding coefficient of $z^{p}$ when expanding $%
Q^{-}(u)|_{\mathcal{F}_{k}^{n}}$ as a power series in $z$, is a well-defined
operator, since we have now limited the number of round trips a particle can
make and there are only a finite number of particles in the system.

\begin{corollary}
In the Fock space representation we have that the Bethe vectors (\ref%
{Bethe_vector}) form a common eigenbasis of $\{T_{r},Q_{r}^{\pm }\}$ with%
\begin{equation}
Q^{+}(u)|y_{\lambda }\rangle =\tprod_{j=1}^{\lambda
_{1}}(1-u~y_{j})~|y_{\lambda }\rangle \;.  \label{Q+spec}
\end{equation}%
In particular, the eigenvalues of $Q_{r}^{+}$ are the elementary symmetric
functions in the Bethe roots $y_{i}$. The eigenvalues of $Q_{r}^{-}$ are
then derived via (\ref{qWronskian}).
\end{corollary}

\begin{proof}
From (\ref{Q2T}) it follows that the Bethe vectors (\ref{Bethe_vector}) are
eigenvectors of $Q_{r}^{\pm }$ as the latter are polynomial in the $T_{r}$%
's. The identity (\ref{TQ+}) together with (\ref{specT}) gives the
eigenvalue $\mathcal{Q}^{+}(u,y_{\lambda })=\tprod_{j=1}^{\lambda
_{1}}(1-u~y_{j})$ in (\ref{Q+spec}).
\end{proof}

Naively one might expect that $Q^{-}(u)$ has analytic eigenvalues as well.
Let us evaluate $u$ in a neighbourhood $U\subset \mathbb{C}$ of $u=0$ and
for a fixed number of lattice sites $n$ and particle number $k=\lambda _{1}$
set $z$ in (\ref{T}) to $z=\varepsilon (n,\lambda _{1})$ with $0<\varepsilon
\ll 1$. If one can show that the following expression converges,%
\begin{equation}
F(u,y_{\lambda })=\mathcal{Q}^{+}(u,y_{\lambda })\sum_{r\geq 0}\frac{%
\varepsilon ^{r}u^{rn}q^{nr(r+1)}q^{2r\lambda _{1}}}{\mathcal{Q}%
^{+}(uq^{2r},y_{\lambda })\mathcal{Q}^{+}(uq^{2r+2},y_{\lambda })}
\label{Q-spec}
\end{equation}%
then the Bethe ansatz equations (\ref{BAE}) imply that the residues at $%
u=q^{-2m}/y_{j}$ vanish and, hence, that $F$ is analytic in $u$. Via
analytic continuation in $\varepsilon $ one then defines $F$ outwith the
region of convergence. The identity (\ref{TQ-}) would then imply that $%
F(u,y_{\lambda })=\mathcal{Q}^{-}(u,y_{\lambda })$\ is the eigenvalue of $%
Q^{-}(u)$. However, the expression (\ref{Q-spec}) is also a power series in $%
q$ with the Bethe roots being potentially Puiseux series in $q$. In order to
control the convergence of the expression (\ref{Q-spec}) one needs to know
the dependence of the Bethe roots on $q$ which is a difficult and technical
issue. We hope to address this problem by different means, which go beyond
the discussion in this article, in future work.

\begin{remark}
\textrm{After the construction of $Q^{+}$ in \cite[Section 3]{CMP13} an
alternative expression for a $Q$-operator has been put forward \cite{Zullo15}
in the Fock representation. In \emph{loc. cit.} the $Q$-operator is stated
in terms of a kernel function for the Jackson integral (Eqns (54) and (62))
which requires special convergence conditions, $0<q<1$, and an upper bound
on the growth of functions in the state space $\mathcal{F}^{n}$. The
derivation of this kernel function in \emph{loc. cit.} postulates the
existence of a $q$-analogue of the $\delta $-function for the Jackson
integral and I was unable to verify whether the constructed operator is
either of the two operators constructed here. }
\end{remark}

N.B. in the proof of (\ref{Q+spec}) we have reversed the usual logic and
derived the spectrum of the $Q^{+}$-operator from the spectrum (\ref{specT})
of the transfer matrix and the functional relation (\ref{TQ+}). As the
system is solvable via the Bethe ansatz and completeness has been proved 
\cite[Section 7]{CMP13}, the significance of the $Q^{+}$-operator in the
present context is \emph{not} that of a mere `auxiliary matrix' as in
Baxter's original work \cite{Baxter07} (and references therein) where it is
used to find the physically relevant eigenvalues of $T$. Instead the $Q^{+}$%
-operator acquires in our setting a \emph{direct} physical significance as
it describes the discretised time evolution (\ref{right_flow}) of the
quantum Ablowitz-Ladik chain, as we will see next.

\section{Quantum B\"{a}cklund transformation}

We define the quantum analogue of the B\"{a}cklund transformation $\mathcal{B%
}^{+}(v)$ in the Fock space representation by setting $\mathcal{B}^{+}(v):%
\limfunc{End}\mathcal{F}^{n}\rightarrow \mathbb{C}[\![v]\!]\otimes \limfunc{%
End}\mathcal{F}^{n}$ with $\mathcal{O}\mapsto Q^{+}(v)\mathcal{O}%
Q^{+}(v)^{-1}$. The quantum B\"{a}cklund map $\mathcal{B}^{+}$ is
well-defined because of the following result.

\begin{proposition}
In the Fock representation the inverse of the $Q^{+}$-operator exists and is
given by%
\begin{equation}
Q^{+}(v)^{-1}=\sum_{r\geq 0}v^{r}\det ((-1)^{1-i+j}Q_{1-i+j}^{+})_{1\leq
i,j\leq r}~,
\end{equation}%
where the coefficients for $r<n$ simplify according to%
\begin{equation}
\det ((-1)^{1-i+j}Q_{1-i+j}^{+})_{1\leq i,j\leq r}=q^{-2r}Q_{r}^{-},\qquad
r<n\;.  \label{Qinverse}
\end{equation}
\end{proposition}

\begin{proof}
The first identity for the inverse operator follows from the Bethe ansatz
result: in the eigenbasis (\ref{Bethe_vector}) the inverse of $Q^{+}(v)$ is
given by the diagonal matrix with entries%
\begin{equation*}
\prod_{j=1}^{\lambda _{1}}\frac{1}{1-vy_{j}}=\sum_{r\geq
0}v^{r}h_{r}(y_{1},\ldots ,y_{\lambda _{1}})
\end{equation*}%
with $h_{r}$ denoting the complete symmetric functions \cite[Ch. I]%
{Macdonald}. Using the known determinant relation between elementary and
complete symmetric functions, $h_{r}=\det (e_{1-i+j})_{1\leq i,j\leq r}$,
the first assertion follows from (\ref{Q+spec}). The second assertion is now
an immediate consequence of (\ref{qWronskian}) when expanding in the
variable $v$.
\end{proof}

We are particularly interested in the image of the $q$-boson algebra
generators under the quantum B\"{a}cklund transform, 
\begin{equation}
\tilde{\beta}_{j}(v)=Q^{+}(v)\beta _{j}Q^{+}(v)^{-1}=\sum_{r\geq 0}v^{r}%
\tilde{\beta}_{j,r}\;.  \label{beta_expansion}
\end{equation}%
We define $\tilde{\beta}_{j}^{\ast }(v)$ and $\tilde{\beta}_{j,r}^{\ast }$
in an analogous fashion. N.B. $\tilde{\beta}_{j}^{\ast }(v)\neq (\tilde{\beta%
}_{j}(v))^{\ast }$ since $Q^{+}$ is not unitary in the Fock space
representation. For ease of notation we will often suppress the explicit
dependence on the variable $v$ but the reader should keep in mind that $%
\tilde{\beta}_{j}(v)$ and $\tilde{\beta}_{j}^{\ast }(v)$ are power series in 
$v$ with coefficients in $\limfunc{End}\mathcal{F}^{n}$ as we are now
working in the Fock space representation. Clearly, the transformed quantum
variables $\{\tilde{\beta}_{j},\tilde{\beta}_{j}^{\ast }\}$ still obey the $%
q $-boson algebra relations and, thus, we obtain a one-variable family of
representations of $\mathcal{H}_{n}$. Moreover, by construction and because
of (\ref{all_commute}) the quantum integrals of motion are left invariant,%
\begin{equation}
T_{r}(\tilde{\beta}_{j},\tilde{\beta}_{j}^{\ast })=Q^{+}(v)T_{r}(\beta
_{j},\beta _{j}^{\ast })Q^{+}(v)^{-1}=T_{r}(\beta _{j},\beta _{j}^{\ast })\;.
\end{equation}

\subsection{The image of the quantum B\"{a}cklund map}

\label{sec:result}

We now derive a set of functional relations for the transformed quantum
variables $\{\tilde{\beta}_{j},\tilde{\beta}_{j}^{\ast }\}$ and show that
they are an \emph{exact} match of the classical relations (\ref{clB}) which
describe the discretised time flow (\ref{right_flow}).

\begin{theorem}
\label{mainresult}The quantum B\"{a}cklund transformed variables $\{\tilde{%
\beta}_{j}(v),\tilde{\beta}_{j}^{\ast }(v)\}$ obey the functional relations 
\begin{equation}
\left\{ 
\begin{array}{c}
\tilde{\beta}_{j}-\beta _{j}=v(1-\beta _{j}^{\ast }\tilde{\beta}_{j})\tilde{%
\beta}_{j-1}\medskip \\ 
\tilde{\beta}_{j}^{\ast }-\beta _{j}^{\ast }=v\beta _{j+1}^{\ast }(\beta
_{j}^{\ast }\tilde{\beta}_{j}-1)%
\end{array}%
\right. \;.  \label{quantumB}
\end{equation}%
These relations determine $\{\tilde{\beta}_{j}(v),\tilde{\beta}_{j}^{\ast
}(v)\}$ via recurrence when expanding in the variable $v$.
\end{theorem}

\begin{proof}
A somewhat lengthy but straightforward computation. We give the intermediate
steps and leave some details to the reader to verify. Set $a_{j}=\beta
_{j}\beta _{j+1}^{\ast }$. Via induction in $m$ one establishes the
identities%
\begin{eqnarray*}
\lbrack a_{j-1}^{m},\beta _{j}] &=&-(1-q^{2m})a_{j-1}^{m-1}\beta
_{j-1}q^{2N_{j}}~, \\
\lbrack a_{j}^{m},\beta _{j}^{\ast }] &=&(1-q^{2m})q^{2N_{j}}\beta
_{j+1}^{\ast }a_{j}^{m-1}\;.
\end{eqnarray*}%
From the explicit expression (\ref{Q+r}) one then derives the recurrence
relations%
\begin{eqnarray*}
\lbrack Q_{r}^{+},\beta _{j}] &=&q^{2N_{j}}Q_{r-1}^{+}\beta _{j-1}-\beta
_{j}^{\ast }[Q_{r-1}^{+},\beta _{j}]\beta _{j-1}~, \\
\lbrack Q_{r}^{+},\beta _{j}^{\ast }] &=&-\beta _{j+1}^{\ast
}Q_{r-1}^{+}q^{2N_{j}}-\beta _{j+1}^{\ast }[Q_{r-1}^{+},\beta _{j}^{\ast
}]\beta _{j}\;.
\end{eqnarray*}%
Multiplying with $v^{r}$ and subsequently summing over $r$ on both sides of
these two equalities produces%
\begin{eqnarray*}
\lbrack Q^{+}(v),\beta _{j}] &=&-vq^{2N_{j}}Q^{+}(v)\beta _{j-1}-v\beta
_{j}^{\ast }[Q^{+}(v),\beta _{j}]\beta _{j-1} \\
\lbrack Q^{+}(v),\beta _{j}^{\ast }] &=&-v\beta _{j+1}^{\ast
}Q^{+}(v)q^{2N_{j}}-v\beta _{j+1}^{\ast }[Q^{+}(v),\beta _{j}^{\ast }]\beta
_{j}\;.
\end{eqnarray*}%
After multiplying with $Q^{+}(v)^{-1}$ from the right we obtain (recall that 
$q^{2N_{j}}=1-\beta _{j}^{\ast }\beta _{j}$)%
\begin{equation*}
\left\{ 
\begin{array}{c}
\tilde{\beta}_{j}-\beta _{j}=v(1-\beta _{j}^{\ast }\beta _{j})\tilde{\beta}%
_{j-1}-v\beta _{j}^{\ast }(\tilde{\beta}_{j}-\beta _{j})\tilde{\beta}%
_{j-1}\medskip \\ 
\tilde{\beta}_{j}^{\ast }-\beta _{j}^{\ast }=v\beta _{j+1}^{\ast }(\tilde{%
\beta}_{j}^{\ast }\tilde{\beta}_{j}-1)-v\beta _{j+1}^{\ast }(\tilde{\beta}%
_{j}^{\ast }-\beta _{j}^{\ast })\tilde{\beta}_{j}%
\end{array}%
\right.
\end{equation*}%
from which the desired equalities in (\ref{quantumB}) now easily follow.
\end{proof}

We demonstrate that (\ref{quantumB}) allows one to compute the image of the $%
q$-boson generators via recurrence. Namely, making an analogous power series
expansion $\tilde{\beta}_{j}(v)=\sum_{r\geq 0}v^{r}\tilde{\beta}_{j,r}$ as
in the classical case (\ref{psi_expansion}) we find for the first few
coefficients $\tilde{\beta}_{j,0}=\beta _{j}$ and%
\begin{eqnarray*}
\tilde{\beta}_{j,1} &=&\beta _{j-1}(1-\beta _{j}^{\ast }\beta _{j}) \\
\tilde{\beta}_{j,2} &=&\beta _{j-2}(1-\beta _{j-1}^{\ast }\beta
_{j-1})(1-\beta _{j}^{\ast }\beta _{j})-\beta _{j-1}^{2}\beta _{j}^{\ast
}(1-\beta _{j}^{\ast }\beta _{j}) \\
\tilde{\beta}_{j,3} &=&\beta _{j-3}(1-\beta _{j-2}^{\ast }\beta
_{j-2})(1-\beta _{j-1}^{\ast }\beta _{j-1})(1-\beta _{j}^{\ast }\beta _{j})
\\
&&\quad -2\beta _{j-2}\beta _{j-1}\beta _{j}^{\ast }(1-\beta _{j-1}^{\ast
}\beta _{j-1})(1-\beta _{j}^{\ast }\beta _{j})\;.
\end{eqnarray*}%
Similarly, we find for $\tilde{\beta}_{j}^{\ast }$ that $\tilde{\beta}%
_{j,0}^{\ast }=\beta _{j}^{\ast }$ and 
\begin{eqnarray*}
\tilde{\beta}_{j,1}^{\ast } &=&-(1-\beta _{j}^{\ast }\beta _{j})\beta
_{j+1}^{\ast } \\
\tilde{\beta}_{j,2}^{\ast } &=&\beta _{j-1}\beta _{j}^{\ast }\beta
_{j+1}^{\ast }(1-\beta _{j}^{\ast }\beta _{j}) \\
\tilde{\beta}_{j,3}^{\ast } &=&-\beta _{j-1}^{2}\beta _{j}^{\ast ~2}\beta
_{j+1}^{\ast }(1-\beta _{j}^{\ast }\beta _{j})+\beta _{j-2}\beta _{j}^{\ast
}\beta _{j+1}^{\ast }(1-\beta _{j-1}^{\ast }\beta _{j-1})(1-\beta _{j}^{\ast
}\beta _{j})\;.
\end{eqnarray*}%
We can interpret the above formulae as describing discrete time steps under
the evolution (\ref{right_flow}): since the quantum relations (\ref{quantumB}%
) and classical relations (\ref{clB}) are the same, these expressions yield
the corresponding coefficients $\tilde{\psi}_{j,r}$ and $\tilde{\psi}%
_{j,r}^{\ast }$ in (\ref{recurrence}) upon replacing $\tilde{\beta}%
_{j,r}\rightarrow \tilde{\psi}_{j,r}$ and $\tilde{\beta}_{j,r}^{\ast
}\rightarrow \tilde{\psi}_{j,r}^{\ast }$.


Naturally, one wants to extend the discussion to include the $Q^{-}$%
-operator and one then needs to show that its inverse exists. However, when
proceeding along the same lines as for the $Q^{+}$-operator one must first
prove convergence of the supposed eigenvalues (\ref{Q-spec}). Since the
latter is currently an open question, we state the following preliminary
result.

\begin{lemma}
We have the following commutation relations between the generators of the $q$%
-boson algebra and $Q^{-}(u)$%
\begin{equation}
\left\{ 
\begin{array}{c}
Q^{-}(v)\beta _{j}-\beta _{j}Q^{-}(v)=-v\beta _{j-1}\left[ Q^{-}(v)+\beta
_{j}Q^{-}(v)\beta _{j}^{\ast }\right] \\ 
Q^{-}(v)\beta _{j}^{\ast }-\beta _{j}^{\ast }Q^{-}(v)=v\left[ Q^{-}(v)-\beta
_{j}Q^{-}(v)\beta _{j}^{\ast }\right] \beta _{j+1}^{\ast }%
\end{array}%
\right.  \label{pre_quantumB-}
\end{equation}
\end{lemma}

\begin{proof}
Via a similar computation as in the case of $Q^{+}$ one shows the
commutation relations%
\begin{eqnarray*}
\lbrack Q_{r}^{-},\beta _{j}] &=&-\beta _{j-1}Q_{r-1}^{-}q^{2N_{j}+2}-\beta
_{j-1}[Q_{r-1}^{-},\beta _{j}]\beta _{j}^{\ast }~, \\
\lbrack Q_{r}^{-},\beta _{j}^{\ast }] &=&q^{2N_{j}+2}Q_{r-1}^{-}\beta
_{j+1}^{\ast }-\beta _{j}[Q_{r-1}^{-},\beta _{j}^{\ast }]\beta _{j+1}^{\ast
}\;.
\end{eqnarray*}%
after multiplying with $v^{r}$ and summing over $r$ the assertion follows.
\end{proof}

Define a one-variable family of operators $\{\hat{\beta}_{j}(v),\hat{\beta}%
_{j}^{\ast }(v)\}$ via the functional relations (compare with (\ref{clB-}))%
\begin{equation}
\left\{ 
\begin{array}{c}
\hat{\beta}_{j}-\beta _{j}=-v\beta _{j-1}(1-\beta _{j}\hat{\beta}_{j}^{\ast
})\medskip \\ 
\hat{\beta}_{j}^{\ast }-\beta _{j}^{\ast }=v~(1-\beta _{j}\hat{\beta}%
_{j}^{\ast })\hat{\beta}_{j+1}^{\ast }%
\end{array}%
\right. \;.  \label{quantumB-}
\end{equation}%
As in the previous case (\ref{quantumB}) the relations (\ref{quantumB-})
also determine $\{\hat{\beta}_{j}(v),\hat{\beta}_{j}^{\ast }(v)\}$ via
recurrence when making the analogous power series expansions $\hat{\beta}%
_{j}(v)=\sum_{r\geq 0}v^{r}\hat{\beta}_{j,r}$. Under the assumption that $%
Q^{-}(v)^{-1}$ exists, it then follows from (\ref{pre_quantumB-}) that $\hat{%
\beta}_{j}(v)=Q^{-}(v)\beta _{j}Q^{-}(v)^{-1}$ and $\hat{\beta}_{j}^{\ast
}(v)=Q^{-}(v)\beta _{j}^{\ast }Q^{-}(v)^{-1}$.

\section{From B\"{a}cklund transformations to 2D TQFT}

We now link the quantum B\"{a}cklund transform (\ref{quantumB}) to the 2D
TQFT constructed in \cite[Section 7]{CMP13}. Consider multivariate
deformations of the $q$-boson algebra by setting 
\begin{equation}
\tilde{\beta}_{j}(x_{1},\ldots ,x_{\ell }):=\left( \tprod_{i=1}^{\ell
}Q^{+}(x_{i})\right) \beta _{j}\left( \tprod_{i=1}^{\ell
}Q^{+}(x_{i})^{-1}\right)  \label{multiB}
\end{equation}%
with $\ell \leq n$. We define $\tilde{\beta}_{j}^{\ast }(x_{1},\ldots
,x_{\ell })$ analogously. We shall concentrate on the $Q^{+}$-operator
because both solutions $Q^{\pm }$ are related via the functional relation (%
\ref{qWronskian})\emph{, }which shows that the transformation induced by $%
Q^{-}$ can be constructed from $Q^{+}$. Since the $Q^{+}(x_{i})$'s in (\ref%
{multiB}) commute with each other - reflecting the analogous property (\ref%
{Bcomm}) of the classical B\"{a}cklund transform - one can expand the
product of $Q^{+}$-operators in any basis in the ring of symmetric functions
in the variables $x_{i}$. Set $P_{\lambda }(x_{1},\ldots ,x_{\ell
};q)=P_{\lambda }(x_{1},\ldots ,x_{\ell };q^{2},0)$ where $P_{\lambda
}(x_{1},\ldots ,x_{\ell };q,t)$ are the celebrated Macdonald functions \cite%
{Macdonald}. Then the expansion%
\begin{equation}
\tprod_{i=1}^{\ell }Q^{+}(x_{i})=\sum_{\lambda }(-1)^{|\lambda |}Q_{\lambda
}P_{\lambda }(x_{1},\ldots ,x_{\ell };q)  \label{fatQ}
\end{equation}%
with the sum running over all partitions with at most $\ell $ parts defines
uniquely a set of commuting polynomials $\{Q_{\lambda }\}$ in the $q$-boson
algebra $\mathcal{H}_{n}\otimes \mathbb{C}[\![z]\!]$, where $z$ is the
\textquotedblleft quasi-periodicity parameter\textquotedblright\ in (\ref{T}%
) which we treat as formal variable. Note that the operators $Q_{\lambda }$
can be defined explicitly as polynomials in the $Q_{r}^{+}$'s, see \cite[Def
3.3]{CMP13}. For the discussion in this article their implicit definition
via (\ref{fatQ}) suffices\footnote{%
The reader should note that the definition of the $Q^{+}$-operator in this
article corresponds to the operator $\boldsymbol{G}^{\prime }(-u)$ in \cite[%
Prop 3.11, Eqn (3.46)]{CMP13} which explains the extra sign factor in (\ref%
{fatQ}) compared to Eqn (3.56) in \emph{loc. cit.}}.

In the Fock representation it follows from the results in \cite[Section 7]%
{CMP13} that if $\lambda $ has a column of height $n$ then $Q_{\lambda
}=z^{m_{n}(\lambda )}Q_{\tilde{\lambda}}$ where $\tilde{\lambda}$ is the
partition obtained from $\lambda $ by removing all columns of height $n$ and 
$m_{n}(\lambda )$ is the multiplicity of these columns. Setting $z=1$
(periodic boundary conditions) we therefore restrict ourselves to $\ell =n-1$%
. We recall the following results from \cite[Thm 7.2, Lem 7.7 and Cor 7.3 ]%
{CMP13}.

\begin{theorem}
(1) The $\{Q_{\tilde{\lambda}}\}$ form a basis in the ring $\mathcal{A}_{n}$
of quantum integrals of motion generated by the $T_{r}$'s. (2) Let $\lambda $%
, $\mu ,~$be partitions with at most $n$ parts and $\mu _{1}=\nu
_{1}=\lambda _{1}=k\in \mathbb{Z}_{\geq 0}$. Then 
\begin{equation}
Q_{\tilde{\lambda}}Q_{\tilde{\mu}}=\sum_{\nu _{1}=k}N_{\tilde{\lambda}\tilde{%
\mu}}^{\tilde{\nu}}(q)Q_{\tilde{\nu}},\qquad N_{\tilde{\lambda}\tilde{\mu}}^{%
\tilde{\nu}}(q)=\langle \nu |Q_{\tilde{\lambda}}|\mu \rangle ~,
\label{N_TQFT}
\end{equation}%
where the expansion coefficients are the fusion coefficient of a 2D TQFT
with $N_{\tilde{\lambda}\tilde{\mu}}^{\tilde{\nu}}(0)$ being the $SU(n)$
WZNW fusion coefficient at level $k$. If $\mu _{1}\neq \nu _{1}$ or $\mu
_{1}\neq \lambda _{1}$ the matrix element is zero.
\end{theorem}

\subsection{A brief summary of 2D TQFT}

For the sake of completeness and to make this article accessible to a wider
audience we briefly summarise the definition of a 2D TQFT, for the precise
definition we refer the interested reader to the abundant literature on the
subject; see e.g. the text book \cite{Kock03}.

In modern mathematical language a 2D TQFT, is a symmetric monoidal functor $%
Z:(\limfunc{2Cob},\sqcup )\rightarrow (\limfunc{Vect}_{\mathbb{\Bbbk }%
},\otimes )$ from the category of 2-cobordisms into the category of
finite-dimensional vector spaces over some base field $\Bbbk $. The objects
in the category $\limfunc{2Cob}$ are circles $\mathbb{S}^{1}$, closed
strings, and we define a product on them via the disjoint union $\sqcup $.
The functor $Z$ maps each closed string $\mathbb{S}^{1}$ onto a vector space 
$Z(\mathbb{S}^{1})$ while preserving the product structure, i.e. $Z(\mathbb{S%
}^{1}\sqcup \cdots \sqcup \mathbb{S}^{1})=Z(\mathbb{S}^{1})\otimes \cdots
\otimes Z(\mathbb{S}^{1})$ where $\otimes $ denotes the ordinary tensor
product of vector spaces. One includes the case where the circle shrinks to
a point $\limfunc{pt}$ and sets $Z(\limfunc{pt})=\Bbbk $. This explains the
monoidal property of $Z$. Asking $Z$ to be \emph{symmetric} refers to the
fact that each circle carries an orientation: let $\mathbb{S}^{1}$ be
anti-clockwise oriented and denote by $\mathbb{\bar{S}}^{1}$ the clockwise
oriented circle. Then one demands that a change in orientation in $\limfunc{%
2Cob}$ corresponds to taking the dual space in $\limfunc{Vect}_{\mathbb{%
\Bbbk }}$, i.e. $Z(\mathbb{\bar{S}}^{1})=Z(\mathbb{S}^{1})^{\ast }$.

\begin{figure}[tbp]
\begin{equation*}
\includegraphics[scale=0.45]{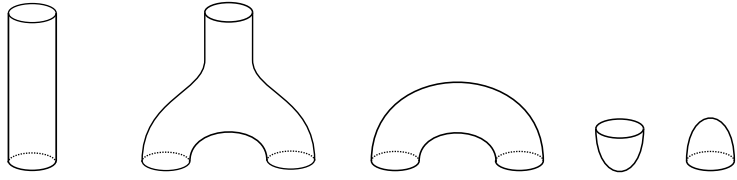}
\end{equation*}%
\caption{Depiction of elementary 2-cobordisms. From left to right: the
identity map $\mathrm{Id}:Z(\mathbb{S}^{1})\rightarrow Z(\mathbb{S}^{1})$,
the `pair of pants' $m:Z(\mathbb{S}^{1})\otimes Z(\mathbb{S}^{1})\rightarrow
Z(\mathbb{S}^{1})$, the `elbow' $:Z(\mathbb{S}^{1})\otimes Z(\mathbb{S}%
^{1})\rightarrow\Bbbk$, the `cup' ${\mathbf{1}}:\Bbbk\rightarrow Z(\mathbb{S}%
^{1})$ and the `cap' $\mathrm{tr}:Z(\mathbb{S}^{1})\rightarrow \Bbbk$. }
\label{fig:2cob}
\end{figure}

The morphisms between objects in $\limfunc{2Cob}$ are two-dimensional
compact manifolds which interpolate between two sets of circles, one
representing the \textquotedblleft in-states\textquotedblright\ $\underset{m}%
{\underbrace{\mathbb{S}^{1}\sqcup \cdots \sqcup \mathbb{S}^{1}}}$ and the
other the \textquotedblleft out-states\textquotedblright\ $\underset{n}{%
\underbrace{\mathbb{S}^{1}\sqcup \cdots \sqcup \mathbb{S}^{1}}}$\ in physics
terminology. The morphisms in $\limfunc{2Cob}$ are called 2-cobordisms and
can be thought of as a 2-dimensional analogue of Feynman diagrams with $m$
incoming particles and $n$ outgoing ones. Each such 2-cobordism is mapped
under $Z$ to an element in $\limfunc{Hom}(Z(\mathbb{S}^{1})^{\otimes m},Z(%
\mathbb{S}^{1})^{\otimes n})$, i.e. a linear map between two tensor products
of the vector space $Z(\mathbb{S}^{1})$.

Because of the functorial property $Z$ is fixed by specifying its values on
certain \textquotedblleft elementary\textquotedblright\ 2-cobordisms shown
in Figure \ref{fig:2cob} which introduce on $Z(\mathbb{S}^{1})$ a product $%
m:Z(\mathbb{S}^{1})\otimes Z(\mathbb{S}^{1})\rightarrow Z(\mathbb{S}^{1})$,
an identity element $\mathbf{1}:\Bbbk \rightarrow Z(\mathbb{S}^{1})$ and an
invariant bilinear form $\langle \cdot |\cdot \rangle :Z(\mathbb{S}%
^{1})\otimes Z(\mathbb{S}^{1})\rightarrow \Bbbk $. Instead of the latter,
one often considers the trace functional given by $\limfunc{tr}(\cdot
)=\langle \boldsymbol{1}|\cdot \rangle $ which defines a map $Z(\mathbb{S}%
^{1})\rightarrow \Bbbk $. These maps endow $Z(\mathbb{S}^{1})$ with the
structure of a \emph{Frobenius algebra }and the latter is called \emph{%
symmetric} if the product is commutative.

\subsection{From $Q$-operators to TQFT fusion matrices}

Hence, in order to prove that (\ref{N_TQFT}) defines a 2D TQFT we need to
endow the algebra $\mathcal{A}_{n}$ generated from the quantum integrals of
motion (\ref{T_r}) with the structure of a symmetric Frobenius algebra. From
(\ref{T_r}) and (\ref{Q+r}), (\ref{Q-r}) we infer that $\mathcal{A}_{n}$
block decomposes over the subspaces $\mathcal{F}_{k}^{n}$ of fixed particle
number,%
\begin{equation*}
\mathcal{A}_{n}=\tbigoplus_{k\geq 0}\mathcal{A}_{n,k},\qquad \mathcal{A}%
_{n,k}\subset \limfunc{End}\mathcal{F}_{k}^{n}
\end{equation*}%
with $\dim \mathcal{F}_{k}^{n}=\dim \mathcal{A}_{n,k}=\binom{n+k-1}{k}$. Fix 
$k\in \mathbb{N}$ and set $\Bbbk =\mathbb{C}\{\!\{q\}\!\}$, the field of
Puiseux series in $q$. Consider the vector space $Z_{k}(\mathbb{S}^{1})=%
\mathcal{A}_{n,k}$ together with the matrix product (\ref{N_TQFT}) in $%
\mathcal{A}_{n,k}\subset \limfunc{End}\mathcal{F}_{k}^{n}$ and 
\begin{equation}
\mathbf{1}_{k}:=Q_{\emptyset }\qquad \text{and}\qquad \langle Q_{\tilde{%
\lambda}}|Q_{\tilde{\mu}}\rangle =\delta _{\tilde{\lambda}^{\ast }\tilde{\mu}%
}/b_{\tilde{\lambda}}(q)\ ,  \label{Frob_def}
\end{equation}%
where $\tilde{\lambda}^{\ast }=(k-\tilde{\lambda}_{n-1},\ldots ,k-\tilde{%
\lambda}_{2},k-\tilde{\lambda}_{1})$ and $b_{\tilde{\lambda}}=\prod_{j\geq
1}(q^{2})_{m_{j}(\tilde{\lambda})}$ with $m_{j}(\tilde{\lambda})$ being the
multiplicity of columns of height $j$ in the reduced partition $\tilde{%
\lambda}$.

\begin{theorem}[{\protect\cite[Thm 7.2 \& Cor 7.3]{CMP13}}]
The maps (\ref{Frob_def}) together with (\ref{N_TQFT}) define a commutative
Frobenius algebra, that is $Z_{k}:(\limfunc{2Cob},\sqcup )\rightarrow (%
\limfunc{Vect}_{\Bbbk },\otimes )$ with $Z_{k}(\mathbb{S}^{1})=\mathcal{A}%
_{n,k}$ is a symmetric monoidal functor.
\end{theorem}

\begin{remark}
\textrm{The fusion coefficients (\ref{N_TQFT}) can be explicitly computed
using the algorithm in \cite[Section 7.3.1]{CMP13}. The latter is based on
an algebra isomorphism which maps $\mathcal{A}_{n,k}$ on a quotient of the
spherical Hecke algebra; see \cite[Thm 7.3]{CMP13} for details.
Alternatively, one can extract $N_{\tilde{\lambda}\tilde{\mu}}^{\tilde{\nu}%
}(q)$ from the combinatorial definition of cylindric generalisations of skew 
$q$-Whittaker functions: taking matrix elements in (\ref{fatQ}) one obtains
symmetric functions $P_{\lambda /d/\mu }$ via%
\begin{equation*}
\sum_{d\geq 0}z^{d}P_{\lambda /d/\mu }(x_{1},\ldots ,x_{n-1};q)=\langle
\lambda |\tprod_{i=1}^{n-1}Q^{+}(x_{i})|\mu \rangle ~.
\end{equation*}%
Here $\lambda /d/\mu $ denotes a cylindric shape, a periodically continued
skew diagram, and each $P_{\lambda /d/\mu }$ equals a sum over $q$-weighted
cylindric semi-standard tableaux; see \cite[Section 6.2]{CMP13}. }
\end{remark}

\subsection{The quantum B\"{a}cklund transform in terms of fusion}

The following result connects the fusion coefficients (\ref{N_TQFT}) for the
simplest, one-variable case, with the quantum B\"{a}cklund transform (\ref%
{quantumB}).

\begin{proposition}
\label{prop:linear} The solution of the quantum B\"{a}cklund transform (\ref%
{quantumB}) satisfy the following linear systems involving the fusion
coefficients (\ref{N_TQFT}),%
\begin{eqnarray}
\sum_{a+b=r}\sum_{\rho }(-1)^{a}N_{(a)\tilde{\mu}}^{\tilde{\rho}}(q)~\langle
\lambda |\tilde{\beta}_{j,b}|\rho \rangle &=&N_{(r)\widetilde{\beta _{j}\mu }%
}^{\tilde{\lambda}}(q)  \label{fusion1} \\
\sum_{a+b=r}\sum_{\rho }\frac{(-1)^{a}N_{(a)\tilde{\mu}}^{\tilde{\rho}}(q)}{%
1-q^{2m_{j}(\mu )+2}}~\langle \lambda |\tilde{\beta}_{j,b}^{\ast }|\rho
\rangle &=&N_{(r)\widetilde{\beta _{j}^{\ast }\mu }}^{\tilde{\lambda}}(q)\ ,
\label{fusion2}
\end{eqnarray}%
where the notation $\beta _{j}\mu $ and $\beta _{j}^{\ast }\mu $ stand for
the partitions obtained by respectively deleting and inserting a column of
height $j$ in the Young diagram of $\mu $.
\end{proposition}

\begin{proof}
This is an immediate consequence of the definition (\ref{fatQ}) of the
fusion matrices and (\ref{N_TQFT}) which yields $N_{(r)\tilde{\mu}}^{\tilde{%
\lambda}}=(-1)^{r}\langle \lambda |Q_{r}^{+}|\mu \rangle $. On the other
hand the definition (\ref{beta_expansion}) implies $\tilde{\beta}%
_{j}(v)Q^{+}(v)=Q^{+}(v)\beta _{j}$. Taking matrix elements in the latter
equality and then expanding in the variable $v$ yields (\ref{fusion1}). The
proof of (\ref{fusion2}) follows along the same lines observing that $\beta
_{j}^{\ast }|\mu \rangle =(1-q^{2m_{j}(\mu )+2})|\beta _{j}^{\ast }\mu
\rangle $ according to (\ref{can_basis}).
\end{proof}

For small particle numbers $k$ there is a third way to compute the fusion
coefficients (\ref{N_TQFT}) using the following \textquotedblleft
operator-state correspondence\textquotedblright\ in the Fock space \cite[Cor
7.3]{CMP13}: let $|k^{n}\rangle =\frac{(\beta _{n}^{\ast })^{k}}{(q)_{k}}%
|0\rangle $ denote the basis vector (\ref{can_basis}) with $k$ particles
sitting at site $n$ and all other sites being empty. Then the basis vectors (%
\ref{can_basis}) with $\lambda _{1}=k$ are obtained by $|\lambda \rangle =Q_{%
\tilde{\lambda}}|k^{n}\rangle $, where $\tilde{\lambda}$ is the partition
obtained from $\lambda $ by removing all columns of height $n$. Thus, we can
introduce a fusion product on the $k$-particle space $\mathcal{F}_{k}^{n}$
by setting $|\lambda \rangle \ast |\mu \rangle =Q_{\tilde{\lambda}}|\mu
\rangle $ and the resulting algebra is isomorphic to the TQFT (\ref{N_TQFT}%
). In order to obtain the fusion coefficient we then need to explicitly
compute $Q_{\tilde{\lambda}}$ as a polynomial in the $Q_{r}^{+}$'s and let
them act on the state vector $|\mu \rangle $. In the final step one then
removes all columns of height $n$ in the partitions in the expansion of $Q_{%
\tilde{\lambda}}|\mu \rangle $ as these correspond to the trivial
representation of $SU(n)$. We demonstrate this on a small example.

\begin{example}
\textrm{Set $n=3$. We wish to relate the fusion coefficients (\ref{N_TQFT})
for particle numbers $k=2,3$ via the relation (\ref{fusion2}). In the latter
equation choose $j=1$ and $\mu =(2,2,1)=\Yvcentermath1{\tiny \yng(2,2,1)}\;$%
. Thus, $\beta _{1}^{\ast }\mu =\Yvcentermath1{\tiny \yng(3,2,1)}$ and the
reduced partition with columns of height $n=3$ deleted is simply $\tilde{\mu}%
=\Yvcentermath1{\tiny \yng(1,1)}\;$. The state $|\mu \rangle $ is a
2-particle state with one particle sitting at site 2 and the other at site
3. The state $\beta _{1}^{\ast }|\mu \rangle =(1-q^{2})|3,2,1\rangle $ is a
3-particle state, with one particle sitting at each site. According to (\ref%
{fusion2}) we have the identity 
\begin{equation}
(1-q^{2})N_{(2)(2,1)}^{\tilde{\lambda}}=\langle \lambda |\tilde{\beta}%
_{1,2}^{\ast }|\mu \rangle -\sum_{\rho }\langle \lambda |\tilde{\beta}%
_{1,1}^{\ast }|\rho \rangle N_{(1)\tilde{\mu}}^{\tilde{\rho}}+\sum_{\rho
}\langle \lambda |\tilde{\beta}_{1,0}^{\ast }|\rho \rangle N_{(2)\tilde{\mu}%
}^{\tilde{\rho}}  \label{Ex}
\end{equation}%
with $\tilde{\beta}_{1,0}^{\ast }=\beta _{1}^{\ast }$, $\tilde{\beta}%
_{1,1}^{\ast }=-\beta _{2}^{\ast }q^{2N_{1}}$ and $\tilde{\beta}_{1,2}^{\ast
}=\beta _{1}^{\ast }\beta _{2}^{\ast }\beta _{3}q^{2N_{1}}$. The first term
on the right hand side in the above equation gives 
\begin{equation*}
\tilde{\beta}_{1,2}^{\ast }|\mu \rangle =\beta _{1}^{\ast }\beta _{2}^{\ast
}\beta _{3}q^{2N_{1}}|\,\Yvcentermath1{\tiny \yng(2,2,1)}\,\rangle
\;=(1-q^{2})(1-q^{4})|\,\Yvcentermath1{\tiny \yng(3,2)}\,\rangle \;.
\end{equation*}%
That is, the first matrix element in (\ref{Ex}) is nonzero only for $\lambda
=(3,2)$. For the other terms we need to compute the 2-particle fusion
coefficients $N_{(1)\tilde{\mu}}^{\tilde{\rho}}$ and $N_{(2)\tilde{\mu}}^{%
\tilde{\rho}}$ first. This can be done using the operator-state
correspondence $|\lambda \rangle =Q_{\tilde{\lambda}}|3,3\rangle $ or the
algorithm in \cite{CMP13}. One finds (we simply write the Young diagrams of
the partitions instead of the $Q_{\tilde{\lambda}}$'s) 
\begin{equation*}
\Yvcentermath1{\tiny \yng(2)\;\ast \;\yng(1,1)}=\Yvcentermath1{\tiny \yng(1)}%
\qquad \text{and}\qquad \Yvcentermath1{\tiny \yng(1)\;\ast \;\yng(1,1)}=%
\Yvcentermath1\,{\tiny \yng(2,1)}\;+(1+q^{2})~\emptyset \;.
\end{equation*}%
Setting $q=0$ one can verify that these give the correct WZNW fusion
coefficients of $SU(3)$ at level $k=2$. Thus, we obtain%
\begin{eqnarray*}
\sum_{\rho }\tilde{\beta}_{1,1}^{\ast }|\rho \rangle N_{(1)\tilde{\mu}}^{%
\tilde{\rho}} &=&-(1-q^{2})(1+q^{2})|\,\Yvcentermath1{\tiny \yng(3,3,2)}%
\,\rangle -(1-q^{4})q^{2}|\,\Yvcentermath1{\tiny \yng(3,2)}\,\rangle \\
\sum_{\rho }\tilde{\beta}_{1,0}^{\ast }|\rho \rangle N_{(2)\tilde{\mu}}^{%
\tilde{\rho}} &=&(1-q^{4})|\,\Yvcentermath1{\tiny \yng(3,1,1)}\,\rangle
\end{eqnarray*}%
for the second and third term on the right hand side of (\ref{Ex}). Adding
all three terms we obtain the 3-particle fusion coefficients $N_{(2)(2,1)}^{%
\tilde{\lambda}}$ and, thus, arrive at the 3-particle fusion product (after
deleting columns of height $n=3$ in all partitions) 
\begin{equation*}
\Yvcentermath1{\tiny \yng(2)\;\ast \;\yng(2,1)}=(1+q^{2})~\left( %
\Yvcentermath1\,{\tiny \yng(1)}\;+~\Yvcentermath1\,{\tiny \yng(1,1)}\;+~%
\Yvcentermath1{\tiny \yng(3,2)}\right)
\end{equation*}%
Again one can check that in the limit $q=0$ these are the correct $SU(3)$%
-WZNW fusion coefficients at level $k=3$. As demonstrated on this simple
example one can use the linear system (\ref{fusion2}) to compute fusion
coefficients recursively. }
\end{example}

\section{Conclusions}

In this article we have added new insight to the quantisation of the
Ablowitz-Ladik chain (\ref{ALH}) by using Baxter's concept of the $Q$%
-operator.

\subsection{Summary of results}

The classical system (\ref{ALH}) possesses a set of integral of motions, a
Poisson commutative subalgebra $\mathcal{A}_{n}^{\hbar =0}$, generated by
the trace of the Lax operator (\ref{classicL}), $\{T_{r},T_{r^{\prime }}\}=0$%
. These integrals of motion generate additional \textquotedblleft
elementary\textquotedblright\ flows which can be used to construct the
physical time flow given by the Hamiltonian (\ref{clH}). Suris constructed
in \cite{Suris97} discrete approximations of these elementary flows in terms
of B\"{a}cklund transformations given in terms of a set of functional
relations (\ref{clB}), (\ref{clB-}). In \cite{Kulish81} the $q$-deformed
oscillator or $q$-boson algebra $\mathcal{H}_{n}$ was considered as a
quantisation of the Poisson algebra (\ref{Poisson}) underlying the classical
chain.

In this article we have discussed the quantum analogues of the classical B%
\"{a}cklund transformations using the $q$-boson algebra. We have stated the
explicit construction of two infinite families $\{Q_{r}^{\pm }\}_{r\in 
\mathbb{N}}$ of elements in the $q$-boson algebra, see (\ref{Q+r}) and (\ref%
{Q-r}), which commute among themselves and with the quantum integrals of
motion $\mathcal{A}_{n}$ generated from the higher Hamiltonians $%
\{T_{r}\}_{r=1}^{n}$ of the quantised Ablowitz-Ladik chain. Introducing the
associated current operators $Q^{\pm }(v)$ we derived their algebraic
dependence in terms of a set of functional equations: Baxter's $TQ$-equation
(\ref{TQ+}) and (\ref{TQ-}), which is a second order difference equation,
and the related Wronskian (\ref{qWronskian}) showing that $Q^{-}$ is closely
related to the inverse of $Q^{+}$ in the limit of large lattice sites. We
showed that these relations hold on the level of the $q$-boson algebra
without the need to specify a representation. Once we fixed the Fock space
representation (\ref{Fock}) we were able to show in addition that $Q^{+}(v)$
is invertible for the periodic chain by employing the completeness of the
Bethe ansatz for the $q$-boson model proved independently in \cite%
{vanDiejen06}\ for $-1<q<1$ and in \cite{CMP13} for $q$ generic. This
allowed us to consider the one-variable family of similarity transformations 
$\mathcal{O}\mapsto Q^{+}(v)\mathcal{O}Q^{+}(v)^{-1}$. Applying these
transformations to the generators of the (noncommutative) $q$-boson algebra (%
\ref{qbosondef}) we have shown that their image obey a set of functional
relations (\ref{quantumB}) which match exactly the relations (\ref{clB}) of
the classical B\"{a}cklund transformation. This is an unexpected result: the
discrete time flow (\ref{right_flow}) of the classical and the quantum
variables is the same.\medskip

An open question is the extension of this result to the $Q^{-}(u)$-operator
for which further investigations are needed to establish the existence of
its inverse. As a preliminary result we showed that assuming the existence
of $Q^{-}(v)^{-1}$ a second set of functional relations (\ref{quantumB-})
follows which again match the classical relations (\ref{clB-}) associated
with the inverse of the classical B\"{a}cklund transformation.

In the final section we then considered multivariate versions of the B\"{a}%
cklund transformation via composition. The novel result is that they
generate the fusion matrices of a 2D TQFT constructed in \cite{CMP13}. This
TQFT is a $q$-deformed version of the $SU(n)_{k}$-WZNW fusion ring or
Verlinde algebra which is recovered in the strong coupling limit $%
q\rightarrow 0$; the related combinatorial description of the WZNW fusion
ring and its relation to quantum cohomology can be found in \cite{KS10} and 
\cite{Korff10}. Because of this limit, it was suggested in \cite[Section 8]%
{CMP13} that this TQFT might be related to the one considered by Teleman 
\cite{Teleman} in the context of twisted $K$-theory.

\subsection{Integrability and topological quantum field theory}

We explain how our findings differ from previously known connections between
classical and quantum integrable systems and TQFTs starting with the quantum
case.

Nekrasov and Shatashvili developed in a series of works the following
connection between TQFTs and the Bethe ansatz equations of quantum
integrable models \cite{NS10}: starting from a 4D Yang-Mills theory one can
consider its vacua in the infrared limit. The moduli space of vacua is
described by a TQFT in terms of a set of equations matching the Bethe ansatz
equations of a quantum integrable model using the Yang-Yang functional. For
the $q$-boson model this approach has been followed recently by string
theorists in \cite{OY14} and \cite{GP15} who put forward generalised
versions of WZNW and Chern-Simons field theories.

In contrast our construction of a TQFT using Baxter's idea of a $Q$-matrix
gives an \emph{operator construction} of the TQFT in terms of $q$-difference
operators, so-called quantum $D$-modules. That this operator construction of
a TQFT is identical to the one obtained via the Bethe ansatz equations is a
non-trivial statement: it requires the proof of the existence of an algebra
isomorphism between the operator version of the TQFT (\ref{fatQ}), (\ref%
{N_TQFT}) and the coordinate ring presentation in terms of the Bethe ansatz
equation as well as the completeness of the Bethe ansatz; see \cite[Section
7 and 8]{CMP13}. As a result one obtains an explicit algorithm which allows
one to compute the fusion coefficients, see \cite{Korffproceedings} and \cite%
[Sec. 7.3.1]{CMP13}. We showed in this article that the fusion coefficients
satisfy linear systems (\ref{fusion1}) and (\ref{fusion2}) with the image of
the $q$-boson algebra generators under the quantum B\"{a}cklund transform.

In the context of classical integrable systems the connection with TQFT
rests on the associativity condition of the underlying Frobenius algebra,
the Witten-Dijkgraaf-Verlinde-Verlinde equations; see Dubrovin's lecture
notes \cite{Dubrovin96}. In contrast, associativity is built in from the
start in the construction of the TQFT in \cite{CMP13} as the fusion matrices
are realised as endomorphisms over the Fock space of the $q$-boson algebra,
that is, the product in the TQFT is the ordinary matrix product. The crucial
issue here is the commutativity of the product which follows from the
construction of the commutative subalgebra $\mathcal{A}_{n}$ generated by
the quantum integrals of motion of the quantised Ablowitz-Ladik chain, $%
\{T_{r}\}$ and $\{Q_{r}^{\pm }\}$. In other words, for fixed particle number 
$k$ the 2D TQFT is the quantised Poisson subalgebra of classical integrals
of motion of the Ablowitz-Ladik chain, $\lim_{\hbar \rightarrow 0}\mathcal{A}%
_{n,k}=\mathcal{A}_{n,k}^{\hbar =0}$. In this article we showed that a basis
of the quantum integrals of motions, the fusion matrices of the TQFT, is
obtained from the quantum analogue of B\"{a}cklund transformations.

\subsection{Open questions}

As pointed our earlier in the text, the main physical significance of our
investigation of the $Q^{+}$-operator for the quantum Ablowitz-Ladik chain
is the result that its adjoint action describes the discrete time evolution (%
\ref{right_flow}) of the system. To obtain the full time evolution one needs
in addition to consider the adjoint of the $Q^{+}$-operator. We wish to
exploit this fact in a forthcoming publication where we discuss the discrete
dynamics of the system and connect it to the TQFT.

An interesting question resulting from our discussion is a possible
extension of our construction of quantum B\"{a}cklund transformations to the
continuum limit, where one obtains the quantum nonlinear Schr\"{o}dinger
model, as well as to different boundary conditions; see e.g. \cite{vDE}, 
\cite{WJ-Z} for a discussion of the $q$-boson system with more exotic
boundaries. We hope to address these and related questions in future
work.\bigskip

\noindent \textbf{Acknowledgment}. It is a pleasure to thank Chris Athorne,
Ulrich Kr\"{a}hmer and Jon Nimmo for discussions. I am also grateful to the
referees for their comments which helped to improve the presentation of this
article.

\end{document}